\documentclass[a4paper, 12pt]{article}
\usepackage{relsize}
\usepackage{times,float}
\usepackage{amsmath}
\usepackage{amsthm,thm-restate}
\theoremstyle{definition}
\usepackage{amssymb}
\usepackage{graphicx,epsf}
\usepackage[round,comma]{natbib}
\usepackage{color}
\usepackage{bm}
\usepackage[colorlinks=true,citecolor=blue]{hyperref}
\usepackage[multiple]{footmisc}
\usepackage[subnum]{cases}
\usepackage[english]{babel}
\usepackage[toc,page]{appendix}
\usepackage{lscape}
\usepackage{caption}
\usepackage{subcaption}
\usepackage{enumerate}
\usepackage{booktabs}
\usepackage{multirow}
\usepackage{accents}
\usepackage[textwidth=5cm,textheight=5cm]{geometry}
\usepackage{cancel}

\usepackage{tikz}
\usepackage{xr}
\usepackage{algorithmicx,algorithm}
\usepackage{algpseudocode}

\def\ignore#1{}
\newlength{\defbaselineskip} 
\newcommand{\setlinespacing}[1]%
{\setlength{\baselineskip}{#1 \defbaselineskip}}

\makeatletter
\newcommand*{\rom}[1]{\romannumeral #1}
\newcommand*{\Rom}[1]{\expandafter\@slowromancap\romannumeral #1@}

\newcommand{\n}{\mathbb{N}}

\newtheorem{cor}{Corollary}
\newtheorem{lemma}{Lemma}
\newtheorem*{lemma*}{Lemma}

\newtheorem{theorem}{Theorem}

\newtheorem*{proposition*}{Proposition}
\newtheorem*{theorem*}{Theorem}
\newtheorem*{definition*}{Definition}
\newtheorem{example}{Example}
\newtheorem*{example*}{Example}

\newtheorem*{remark*}{Remark}

\newtheorem*{assumption*}{Assumption}

\newtheorem*{condition*}{Condition}

\newtheorem*{question*}{Question}
\newtheorem{claim}{Claim}
\newtheorem*{claim*}{Claim}

\newtheorem*{step*}{Step}

\newtheorem{fact}{Fact}
\newtheorem*{fact*}{Fact}

\begin{document}

\title{The Crawler: Three Equivalence Results \\ for Object (Re)allocation Problems \\ when Preferences Are Single-peaked\thanks{We are grateful to William Thomson for his invaluable advice and support, and for his patient reviews and detailed suggestions.  We thank Maciej Kotowski for his comments.  We thank seminar participants at the 2019 Ottawa Microeconomic Theory Workshop, University of Rochester, Osaka University, Keio University, and Waseda University.  Finally, we thank an associate editor and anonymous referees for their valuable comments and suggestions.  Tamura is grateful for financial support from Tamkeen under the NYUAD Research Institute award for Project CG005. Also, Tamura thanks the Institute of Social and Economic Research at Osaka University for their summer research support. This paper is based on a chapter of Tamura's Ph.D. dissertation at the University of Rochester. Hosseini was partially supported by NSF grant \#2052488.  All remaining errors are our own.}}
\author{Yuki Tamura\footnote{Center for Behavioral Institutional Design, New York University Abu Dhabi; email:{\tt  yuki.tamura@nyu.edu}} \and Hadi Hosseini\footnote{College of Information Sciences and Technology, Pennsylvania State University; email:{\tt hadi@psu.edu}}}
\date{\today \vskip 0.2cm
}
\maketitle
\begin{abstract}
For object reallocation problems, if preferences are strict but otherwise unrestricted, the Top Trading Cycles rule (TTC) is the leading rule: It is the only rule satisfying efficiency, individual rationality, and strategy-proofness. However, on the subdomain of single-peaked preferences, \cite{bade19} defines a new rule, the “crawler”, which also satisfies these three properties. (i) The crawler selects an allocation by ``visiting" agents in a specific order. A natural ``dual" rule can be defined by proceeding in the reverse order. Our first theorem states that the crawler and its dual are actually the same. (ii) Single-peakedness of a preference profile may in fact hold for more than one order and its reverse. Our second theorem states that the crawler is invariant to the choice of the order. (iii) For object allocation problems (as opposed to reallocation problems), we define a probabilistic version of the crawler by choosing an endowment profile at random according to a uniform distribution, and applying the original definition.  Our third theorem states that this rule is the same as the ``random priority rule".  
\\[2ex]
\end{abstract}


\section{Introduction}

Consider a group of agents each of whom is endowed with an indivisible good, called an ``object".  Each agent has preferences over the objects.  The initial allocation may not be efficient (in the sense of Pareto efficiency) and the issue arises of reallocating the objects so as to achieve efficiency as well as possibly other socially desirable properties. An example of this type of problems is when the agents are households and the objects are housing units they own (the stylized ``housing" market of \cite{ss74}).  A rule is a single-valued mapping that associates with each such problem an allocation, interpreted as a recommendation for the problem. 
If preferences are strict but otherwise unrestricted, the Top Trading Cycles rule (TTC) is the leading rule \citep{ss74}: It is the \textit{only} rule satisfying the three desirable properties of ``efficiency", ``individual rationality",\footnote{Another common name for this property is the ``endowment lower bound".} and ``strategy-proofness" \citep{ma94}.\footnote{Other proofs of this uniqueness result can be found in \cite{svensson99}, \cite{anno15}, \cite{sethuraman16}, and \cite{bade19}.} 

Interestingly, TTC is not the only rule satisfying these properties on the subdomain of ``single-peaked" preferences \citep{bade19}.  
Returning to our example of a housing market, suppose that the housing units are of different sizes and that each household evaluates units based on their size.  A single person may prefer a small unit; a family with children may prefer a large one.  Each household has an ideal size; the further the size of a unit is from this ideal size, in either direction, the less desirable the unit is.  Thus, households have single-peaked preferences with respect to size.  Instead of size, the order could be based on how expensive units are, or on their proximity to a school or to the central business district.  Many other examples can be found where agents have single-peaked preferences with respect to some reference order on the object set.\footnote{Accordingly, this domain has been studied from several other viewpoints (\citealp{liu18}; \citealp{dbcm15}; \citealp{bmrs19}).  Moreover, for other types of resource allocation problems, a single-peakedness is a natural assumption.  An example is when an infinitely divisible commodity has to be fully allocated among a group of agents \citep{sprumont91}.}  

On the single-peaked domain, \cite{bade19} defines a new rule, which she calls the ``crawler", and shows that this rule, as TTC does, satisfies efficiency, individual rationality, and strategy-proofness.    
The idea underlying the crawler is as follows.  Objects are labeled in such a way that preferences are single-peaked with respect to this order.  Similarly, agents are ordered according to their ownership of the objects.  They are visited from left to right and each agent in turn is asked if he most prefers his endowment or an object to the left of his endowment.  If he answers yes, he is asked which object he most prefers. He is assigned that object and leaves.  Otherwise, the agent on his immediate right is asked the same question.  At least one agent has to most prefer his endowment or an object to the left of his endowment.  So some agent is eventually assigned an object, and he leaves with his assignment.  Then the problem is updated as follows.  Consider all of the agents whose endowments are between the object assigned to the agent who left and the object that agent owned (if the agent is not assigned his endowment).  For each of those agents, ownership is shifted by one spot to the right.  
The sweeping process is repeated in the updated problem, and continues until every agent has been assigned an object.  

Obviously, one could define a ``dual" rule by visiting agents from right to left \citep{bade19}.  One would expect the crawler and its dual to differ.  Because each of the two crawlers visits agents in a specific order, one may conjecture that the specific order that is chosen confers a particular advantage to some agents based on the location of their endowments relative to the location of the other agents' endowments.  Our first result is that they are in fact the same (Theorem~\ref{symmetry}).  Thus, whether agents are visited from left to right or from right to left makes no difference.

Moreover, there may be multiple orders over the object set with respect to which a given preference profile is single-peaked.\footnote{\cite{elo08} define an algorithm that identifies all of the orders over an object set with respect to which a preference profile is single-peaked.}  The crawler could select different allocations for these various orders.  Then again, depending on which order we choose, certain agents would be favored at the expense of others.  Our second result is that here too the crawler is in fact invariant to the order over the object set for which single-peakedness holds (Theorem~\ref{order_inv}).  

Another important class of problems is ``object allocation problems"; there, instead of being owned individually, objects are owned collectively.  Well-studied rules for object allocation problems are the ``sequential priority rules"\footnote{Another common name for these rules is ``serial dictatorship".}: To each order on the agent set is associated such a rule: The agent who is first is assigned his most preferred object; the agent who is second is assigned his most preferred object among the remaining ones; and so on.  
We may ask whether the crawler bears some relation to the sequential priority rules.  Because the procedures underlying the definitions of the crawler and of the sequential priority rules are based on completely different considerations, one may doubt that such a relation exists.  Yet, our third theorem provides a positive answer to this question.  It is based on allowing a rule to select a probability distribution over allocations. That is, the rule is probabilistic. 

Given a preference profile, let us select an endowment profile at random according to a uniform distribution, and apply the crawler to the induced object reallocation problem.
We call the probabilistic rule so defined the ``crawler from random endowments".  We prove that the crawler from random endowments is the same probabilistic rule as ``the random priority rule" \citep{as98}: Choose an order on the agent set at random according to a uniform distribution and apply the induced sequential priority rule (Theorem \ref{thm_random}).\footnote{Another common name for the random priority rule is the ``random serial dictatorship".} 

Equivalence results in the same vein already provided in the literature (\citealp{knuth96}; \citealp{as98}; \citealp{ps11}; \citealp{ls11}; \citealp{su05}; \citealp{ekici17}; \citealp{carroll14}; \citealp{bade19equiv}).
Each of them states that a probabilistic version of (a generalized) TTC is equivalent to the random priority rule (or a variant).  However, as we show in Section \ref{dis}, our result cannot be deduced from any of these equivalences.  

This paper is organized as follows.  In Section \ref{model}, we define the model.  We formally define the crawler. Also, we define a dual rule by visiting agents from right to left as opposed to from left to right, and state our first equivalence result: the crawler and this dual rule are the same.  Our second equivalence result is that the crawler is invariant to the order over the object set that preserves single-peakedness.  In Section \ref{prob}, we define the crawler from random endowments, and state our third equivalence result: the crawler from random endowments is the same as the random priority rule.  Proofs are collected in the appendix.

\section{Model}\label{model}
There is a set $N = \{1,2,\hdots,n\}$ of agents and a set $O$ of objects ($|O| = n$). 
Each agent is endowed with one object in $O$, no two agents being endowed with the same object.     
We denote by $\omega$ the endowment profile, i.e., $\omega = (\omega_1,\omega_2,\hdots,\omega_n)$ where the $i^{th}$ coordinate of $\omega$ is the object owned by agent $i$.  
Each agent $i \in N$ has strict preferences $P_i$ over $O$.  Let $\mathcal{P}$ be the set of all preferences.  We write $R_i$ to denote the ``at least as desirable as" relation associated with $P_i$.  That is, for each pair $o,o' \in O$, $o~R_i~o'$ if and only if either $o~P_i~o'$ or $o = o'$.  We represent $P_i$ by an ordered list of the objects, such as
\begin{align*}
P_i:~o,~\tilde{o},~o',\hdots. 
\end{align*}
Let $\mathcal{P}^N$ be the set of preference profiles for $N$.  Our generic notation for a preference profile is $P = (P_i)_{i \in N}$. 

A \textbf{problem} is defined by a preference profile and an endowment profile.
An allocation is a list $x~=~(x_1,x_2,\hdots,x_n)$ such that for each $i \in N$, $x_i \in O$, and for each pair $i,j \in N$ such that $i \neq j$, $x_i \neq x_j$.  
Let $\mathcal{X}$ be the set of allocations.  
A \textbf{rule} is a single-valued mapping $\varphi:\mathcal{P}^N \times \mathcal{X} \rightarrow \mathcal{X}$ that associates with each problem $(P,\omega) \in \mathcal{P}^N \times \mathcal{X}$ an allocation $x \in \mathcal{X}$.  


\bigskip
 
Let $\mathcal{L}$ be the set of strict orders on $O$. 
We consider the following restriction on preference profiles. There is an order $\prec \in \mathcal{L}$ such that for each agent, there is an object with the property that the further with respect to $\prec$ an object is from that distinguished object, in either direction, the worse off he is.  Formally, a preference profile $P$ is \textbf{single-peaked} if there is $\prec \in \mathcal{L}$ such that for each $i \in N$, there is an object, which we denote by $p(P_i)$, such that for each pair $o,o' \in O$, if either $o' \prec o \precsim p(P_i)$ or $p(P_i) \precsim o \prec o'$, then $o~P_i~o'$. 

Throughout, we consider preference profiles that are single-peaked with respect to some reference order on the object set. 
Given $\prec \in \mathcal{L}$, for each $o \in O$, we denote by $o-1$ and $o+1$ the objects that are adjacent to object $o$.\footnote{Obviously, the leftmost object has no object to its left, and the rightmost object has no object to its right.}  
Likewise, given $\prec \in \mathcal{L}$, for each $i \in N$, we denote by $i-1$ and $i+1$ the agents whose endowments are adjacent to that of agent~$i$.\footnote{The agent whose endowment is leftmost has no one to his left, and the agent whose endowment is rightmost has no one to his right.}  Moreover, for each pair $i,j \in N$, if agent $j$ is to the right of agent $i$, we write that $i \prec j$.  

Let $f: \{1,2,\cdots,n\} \rightarrow N$ be a bijection and let $f = (f(1),f(2),\cdots,f(n))$ be the resulting strict order on $N$.  Let $\mathcal{F}$ be the set of all orders on $N$. 
For each $i \in N$, each $P_i \in \mathcal{P}$, and each $O' \subseteq O$, let $X_i(O')$ be the most preferred object of agent $i$ in $O'$ at $P_i$, i.e.,
\[X_i(O') = o~\iff~o \in O'~\text{and}~\text{for each}~o' \in O' \backslash \{o\},~o~P_i~o'.\] 

We conclude this section by defining three basic properties of rules.  Recall that $\varphi$ is our generic notation for a rule.  
First, we require that for each problem, the chosen allocation be such that there is no other allocation that all agents find at least as desirable and at least one agent prefers:

\paragraph{Efficiency:} For each $(P,\omega) \in \mathcal{P}^N \times \mathcal{X}$, there is no $x \in \mathcal{X}$ such that (i) for each $i \in N$, $x_i~R_i~\varphi_i(P,\omega)$, and (ii) there is $j \in N$ such that $x_j~P_j~\varphi_j(P,\omega)$. 

\bigskip

Second, we require that for each problem, each agent find his assignment at least as desirable as his endowment:
\paragraph{Individual rationality:} For each $(P,\omega) \in \mathcal{P}^N \times \mathcal{X}$ and each $i \in N$,
\[\varphi_i(P,\omega)~R_i~\omega_i.\]

\bigskip

Third, we require that no agent ever benefit by misrepresenting his preferences:  
\paragraph{Strategy-proofness:} For each $(P,\omega) \in \mathcal{P}^N \times \mathcal{X}$, each $i \in N$, and each $P'_i \in \mathcal{P}$, 
\[\varphi_i((P_i,P_{-i}),\omega)~R_i~\varphi_i((P'_i,P_{-i}),\omega).\]

\subsection{The crawler}

In the context of object reallocation problems with strict preferences, TTC has been the central rule in the literature: It is the only rule satisfying \textit{efficiency}, \textit{individual rationality}, and \textit{strategy-proofness} \citep{ma94}.  However, on the single-peaked domain, TTC is not the only rule satisfying these properties.  \cite{bade19} defines a new rule for this domain, and shows that this rule also satisfies these three properties.    

The idea underlying the rule is as follows. Objects are labeled in such a way that preferences are single-peaked with respect to the order.  Agents are visited according to the way the objects they own are ordered, from left to right, we say ``in ascending order".  Each agent in turn is asked if he most prefers his endowment or an object to the left of his endowment.  If yes, he is further asked which specific object it is.  Depending on his answer, we take one of the following actions:

\noindent(1) If an agent's most preferred object is to the right of his endowment, we move to the next agent.

\noindent(2) If an agent's most preferred object is his endowment or an object to the left of his endowment, he receives his most preferred object and leaves with it. 
 
\noindent When an agent leaves, the problem is updated as follows.  Consider agents whose endowments are  between the object assigned to the agent who left and the object that agent owned.  The ownership of each such agent is shifted by one spot to the right.  Hence, each of these agents owns a new object.  The sweeping procedure is  repeated in the updated problem.  In anticipation of a forthcoming definition, we refer to the rule just defined as the \textbf{ascending crawler}. We denote it by $\bm{ACR}$.  

The formal description is as follows. 

\paragraph{Ascending crawler, $\bm{ACR}$:} 
Let $\prec \in \mathcal{L}$, $P \in \mathcal{P}^N$, where $P$ is single-peaked with respect to $\prec$, and $\omega \in \mathcal{X}$.   

Label the objects in $O^0 \equiv O$ in such a way that for each $t~\in~\{1,\hdots,~n~-~1~\}$, $o_t \prec o_{t+1}$.  Label the agents in $N^0 \equiv N$ in such a way that for each $t \in \{1,\hdots,n-1\}$, we have $\omega_{i_t} \prec \omega_{i_{t+1}}$.
Let $\widehat{O}^0 = \{o_1,\hdots,o_n\}$ and $\widehat{N}^0 = \{i_1,\hdots,i_n\}$.  

At each Step $k \geq 1$, let 
\[k^* \equiv \min_{\{1,\cdots,n-k\}} \left\{t:~o_t~P_{i_t}~o_{t+1},~\text{where}~o_t,o_{t+1} \in \widehat{O}^{k-1}\right\}.\] 
Let $ACR_{i_{k^*}}(P,\omega) = X_{i_{k^*}}(O^{k-1})$. 
Let $O^k \equiv O^{k-1} \backslash \{X_{i_{k^*}}(O^{k-1})\}$ and $N^k \equiv N^{k-1} \backslash \{i_{k^*}\}$.

Label the objects in $O^k$ in such a way that for each $t \in \{1,\hdots,~n-k-1\}$, $o_t \prec o_{t+1}$.  Label the agents in $N^k$ in such a way that for each $t \in \{1,\hdots,n-k-1\}$, at Step $k-1$, we have $i_t \prec i_{t+1}$.  Let $\widehat{O}^k = \{o_1,\hdots,o_{n-k}\}$ and $\widehat{N}^k = \{i_1,\hdots,i_{n-k}\}$. 

The procedure terminates when no agent remains.  Because there are finite numbers of agents, there are finitely many steps.

\begin{example}\label{ex_illu}
\textit{Illustrating the ascending crawler.}
Let $N = \{1,2,3,4\}$.  Let $\omega_1~\prec~\omega_2~\prec~\omega_3~\prec~\omega_4$.
Let $P \in \mathcal{P}^N$ be defined by 
\begin{align*}
P_1&: \omega_4,~\omega_3,~\omega_2,~\omega_1 \\
P_2&: \omega_2,~\omega_1,~\omega_3,~\omega_4 \\
P_3&: \omega_1,~\omega_2,~\omega_3,~\omega_4 \\
P_4&: \omega_2,~\omega_1,~\omega_3,~\omega_4.
\end{align*}
At each step, we ask the following question to the agent we visit;

\begin{centering}
\textit{`Among the available objects, is your most preferred object either your endowment or an object to the left of your endowment? If the answer is yes, which object is it?'}\footnote{For the agent whose endowment is leftmost, we ask `among the available objects, is your most preferred object your endowment?'}
\end{centering}
\begin{itemize}
\item[]\textbf{Step 1}: Agent 1 is queried first.  Because his answer is no, agent 2 is queried next.  Agent 2 answers yes, and that he most prefers his endowment.  Hence, he receives his endowment.  He leaves with it. 
\item[]\textbf{Step 2}: Agent 1 is queried first.  Because his answer is no, agent 3 is queried next.  Agent~3 answers yes, and that he most prefers $\omega_1$.  Hence, he receives that object and leaves with it.  The ownership of agent 1 is shifted by one spot to the right.  
\item[]\textbf{Step 3}: We repeat the sequence of queries.  Agent 4 is the fist agent who answers yes to the first question, and he most prefers object $\omega_3$.  Hence, he receives that object and leaves with it.  The ownership of agent 1 is shifted by one spot to the right.  
\item[]\textbf{Step 4}: Agent 1 is the only agent who has not been assigned an object yet and object $\omega_4$ is the only available object.  Hence, agent 1 receives object $\omega_4$ and leaves with it.  
\end{itemize}
Because no agent remains, the algorithm terminates, yielding 
\begin{align*}
ACR(P,\omega) = (\omega_4,\omega_2,\omega_1,\omega_3).  
\end{align*}
Figure \ref{CRexample} illustrates the process.  At each step, the agent who is assigned an object at that step is circled.  At each step, the agents who have already been assigned objects are shown in boxes.  
\begin{figure}[H]
\centering
\tikzset{every picture/.style={line width=0.75pt}} 

\begin{tikzpicture}[x=0.75pt,y=0.75pt,yscale=-1,xscale=1]

\draw    (100,120) -- (390,120) ;

\draw    (120,115) -- (120,125) ;

\draw    (200,115) -- (200,125) ;

\draw    (280,115) -- (280,125) ;

\draw    (360,115) -- (360,125) ;

\draw    (119,11.33) .. controls (147.67,11.33) and (133.67,20) .. (162.33,22) ;

\draw    (159.67,50) -- (199,42) ;

\draw    (199,42) .. controls (217.67,47.33) and (217,59.33) .. (229,64) ;

\draw    (319.67,42.67) -- (359,31.33) ;

\draw   (189.67,211.67) .. controls (189.67,205.22) and (194.89,200) .. (201.33,200) .. controls (207.78,200) and (213,205.22) .. (213,211.67) .. controls (213,218.11) and (207.78,223.33) .. (201.33,223.33) .. controls (194.89,223.33) and (189.67,218.11) .. (189.67,211.67) -- cycle ;
\draw   (269,260) .. controls (269,253.56) and (274.22,248.33) .. (280.67,248.33) .. controls (287.11,248.33) and (292.33,253.56) .. (292.33,260) .. controls (292.33,266.44) and (287.11,271.67) .. (280.67,271.67) .. controls (274.22,271.67) and (269,266.44) .. (269,260) -- cycle ;
\draw   (191,252) -- (212.67,252) -- (212.67,273.67) -- (191,273.67) -- cycle ;
\draw    (280.67,271.67) .. controls (211.7,299.72) and (173.76,292.16) .. (121.26,272.6) ;
\draw [shift={(119.67,272)}, rotate = 380.56] [color={rgb, 255:red, 0; green, 0; blue, 0 }  ][line width=0.75]    (10.93,-3.29) .. controls (6.95,-1.4) and (3.31,-0.3) .. (0,0) .. controls (3.31,0.3) and (6.95,1.4) .. (10.93,3.29)   ;

\draw  [dash pattern={on 4.5pt off 4.5pt}]  (119.67,250.67) .. controls (178.77,232.94) and (244.33,234.61) .. (276.23,244.86) ;
\draw [shift={(277.67,245.33)}, rotate = 198.8] [color={rgb, 255:red, 0; green, 0; blue, 0 }  ][line width=0.75]    (10.93,-3.29) .. controls (6.95,-1.4) and (3.31,-0.3) .. (0,0) .. controls (3.31,0.3) and (6.95,1.4) .. (10.93,3.29)   ;

\draw   (192.33,308.67) -- (214,308.67) -- (214,330.33) -- (192.33,330.33) -- cycle ;
\draw   (112.33,308.67) -- (134,308.67) -- (134,330.33) -- (112.33,330.33) -- cycle ;
\draw   (350.33,318.33) .. controls (350.33,311.89) and (355.56,306.67) .. (362,306.67) .. controls (368.44,306.67) and (373.67,311.89) .. (373.67,318.33) .. controls (373.67,324.78) and (368.44,330) .. (362,330) .. controls (355.56,330) and (350.33,324.78) .. (350.33,318.33) -- cycle ;
\draw    (362,330) .. controls (337.38,343.11) and (318.89,342.01) .. (283.3,331.15) ;
\draw [shift={(281.67,330.64)}, rotate = 377.18] [color={rgb, 255:red, 0; green, 0; blue, 0 }  ][line width=0.75]    (10.93,-3.29) .. controls (6.95,-1.4) and (3.31,-0.3) .. (0,0) .. controls (3.31,0.3) and (6.95,1.4) .. (10.93,3.29)   ;

\draw  [dash pattern={on 4.5pt off 4.5pt}]  (281.67,303.33) .. controls (311.72,292.86) and (326.41,292.01) .. (357.1,301.39) ;
\draw [shift={(359,301.98)}, rotate = 197.35] [color={rgb, 255:red, 0; green, 0; blue, 0 }  ][line width=0.75]    (10.93,-3.29) .. controls (6.95,-1.4) and (3.31,-0.3) .. (0,0) .. controls (3.31,0.3) and (6.95,1.4) .. (10.93,3.29)   ;

\draw   (191.67,358.67) -- (213.33,358.67) -- (213.33,380.33) -- (191.67,380.33) -- cycle ;
\draw   (111.67,358.67) -- (133.33,358.67) -- (133.33,380.33) -- (111.67,380.33) -- cycle ;
\draw   (349.67,368.33) .. controls (349.67,361.89) and (354.89,356.67) .. (361.33,356.67) .. controls (367.78,356.67) and (373,361.89) .. (373,368.33) .. controls (373,374.78) and (367.78,380) .. (361.33,380) .. controls (354.89,380) and (349.67,374.78) .. (349.67,368.33) -- cycle ;
\draw   (270.33,358.67) -- (292,358.67) -- (292,380.33) -- (270.33,380.33) -- cycle ;

\draw (112,16) node [scale=0.8] [align=left] {$\displaystyle P_{3}$};
\draw (221.33,28) node [scale=0.8] [align=left] {$\displaystyle P_{2} \ P_{4}$};
\draw (373.33,28.33) node [scale=0.8] [align=left] {$\displaystyle P_{1}$};
\draw (66.67,150.33) node [scale=0.8] [align=left] {object};
\draw (120,150) node [scale=0.8] [align=left] {1};
\draw (202,150) node [scale=0.8] [align=left] {2};
\draw (280,148.67) node [scale=0.8] [align=left] {3};
\draw (360.67,149.33) node [scale=0.8] [align=left] {4};
\draw (66.67,177) node [scale=0.8] [align=left] {agent};
\draw (120,176.67) node [scale=0.8] [align=left] {1};
\draw (202,176.67) node [scale=0.8] [align=left] {2};
\draw (280,175.33) node [scale=0.8] [align=left] {3};
\draw (360.67,176) node [scale=0.8] [align=left] {4};
\draw (66.67,212.33) node [scale=0.8] [align=left] {step 1};
\draw (120,212) node [scale=0.8] [align=left] {1};
\draw (202,212) node [scale=0.8] [align=left] {2};
\draw (280,210.67) node [scale=0.8] [align=left] {3};
\draw (360.67,211.33) node [scale=0.8] [align=left] {4};
\draw (67.33,262.67) node [scale=0.8] [align=left] {step 2};
\draw (120.67,262.33) node [scale=0.8] [align=left] {1};
\draw (202.67,262.33) node [scale=0.8] [align=left] {2};
\draw (280.67,261) node [scale=0.8] [align=left] {3};
\draw (361.33,261.67) node [scale=0.8] [align=left] {4};
\draw (66.67,318.33) node [scale=0.8] [align=left] {step 3};
\draw (204,319) node [scale=0.8] [align=left] {2};
\draw (282,317.67) node [scale=0.8] [align=left] {1};
\draw (362.67,318.33) node [scale=0.8] [align=left] {4};
\draw (123.17,319.5) node [scale=0.8] [align=left] {3};
\draw (66,368.33) node [scale=0.8] [align=left] {step 4};
\draw (203.33,369) node [scale=0.8] [align=left] {2};
\draw (362,368.33) node [scale=0.8] [align=left] {1};
\draw (122.5,369.5) node [scale=0.8] [align=left] {3};
\draw (282,369) node [scale=0.8] [align=left] {4};

\end{tikzpicture}
\caption{The ascending crawler.}
\label{CRexample}
\end{figure}
\end{example}

We restate a result due to \cite{bade19}.
\begin{theorem*}
The ascending crawler is \textit{efficient}, \textit{individually rational}, and \textit{strategy-proof}. 
\end{theorem*}

\subsection{The first equivalence result: Equivalence between the ascending crawler and the descending crawler} 

Of course, as pointed out by \cite{bade19}, a ``dual" rule can be defined by visiting agents from right to left, let us say ``in descending order". 
Let us call this dual rule the \textbf{descending crawler}.  We denote it by $\bm{DCR}$. 
We omit the formal definition and only show how to apply it to our Example \ref{ex_illu}. 

\addtocounter{example}{-1}
\begin{example}[Continued]
\textit{Illustrating the descending crawler.}
At each step, we ask the following question to the agent we visit;

\begin{centering}
\textit{`Among the available objects, is your most preferred object either your endowment or an object to the right of your endowment? If the answer is yes, which object is it?'}\footnote{For the agent whose endowment is rightmost, we ask `among the available objects, is your most preferred object your endowment?'}
\end{centering}
\begin{itemize}
\item[]\textbf{Step 1}: Agent 4 is queried first.  Because his answer is no, agent 3 is queried next.  Because his answer is no, agent 2 is queried next.  Agent 2 answers yes, and that he most prefers his endowment.  Hence, he receives his endowment and leaves with it.  
\item[]\textbf{Step 2}: Agent 4 is queried first.  Because his answer is no, agent 3 is queried next.  Because his answer is no, agent 1 is queried next.  Agent 1 answers yes, and that he most prefers object $\omega_4$.  Hence, he receives that object and leaves with it.  The ownership of each of agents 3 and 4 is shifted by one spot to the left.
\item[]\textbf{Step 3}: We repeat the sequence of queries.  Agent 3 is the first agent who answers yes to the first question, and he most prefers object $\omega_1$.  Hence, he receives that object and leaves with it.  
\item[]\textbf{Step 4}: Agent 4 is the only agent who has not been assigned an object yet and object $\omega_3$ is the only available object.  Hence, agent 4 receives object $\omega_3$ and leaves with it.  
\end{itemize}
Because no agent remains, the algorithm terminates, yielding 
\begin{align*}
DCR(P,\omega) = (\omega_4,\omega_2,\omega_1,\omega_3).  
\end{align*}
\begin{figure}[H]
\centering
\tikzset{every picture/.style={line width=0.75pt}} 

\begin{tikzpicture}[x=0.75pt,y=0.75pt,yscale=-1,xscale=1]

\draw    (100,120) -- (390,120) ;

\draw    (120,115) -- (120,125) ;

\draw    (200,115) -- (200,125) ;

\draw    (280,115) -- (280,125) ;

\draw    (360,115) -- (360,125) ;

\draw    (119,11.33) .. controls (147.67,11.33) and (133.67,20) .. (162.33,22) ;

\draw    (159.67,50) -- (199,42) ;

\draw    (199,42) .. controls (217.67,47.33) and (217,59.33) .. (229,64) ;

\draw    (319.67,42.67) -- (359,31.33) ;

\draw   (189.67,211.67) .. controls (189.67,205.22) and (194.89,200) .. (201.33,200) .. controls (207.78,200) and (213,205.22) .. (213,211.67) .. controls (213,218.11) and (207.78,223.33) .. (201.33,223.33) .. controls (194.89,223.33) and (189.67,218.11) .. (189.67,211.67) -- cycle ;
\draw   (108,262.33) .. controls (108,255.89) and (113.22,250.67) .. (119.67,250.67) .. controls (126.11,250.67) and (131.33,255.89) .. (131.33,262.33) .. controls (131.33,268.78) and (126.11,274) .. (119.67,274) .. controls (113.22,274) and (108,268.78) .. (108,262.33) -- cycle ;
\draw   (191,252) -- (212.67,252) -- (212.67,273.67) -- (191,273.67) -- cycle ;
\draw    (119.67,274) .. controls (168.84,302.55) and (311.12,289.76) .. (357.3,273.33) ;
\draw [shift={(358.67,272.83)}, rotate = 519.54] [color={rgb, 255:red, 0; green, 0; blue, 0 }  ][line width=0.75]    (10.93,-3.29) .. controls (6.95,-1.4) and (3.31,-0.3) .. (0,0) .. controls (3.31,0.3) and (6.95,1.4) .. (10.93,3.29)   ;

\draw  [dash pattern={on 4.5pt off 4.5pt}]  (278,246.83) .. controls (241.7,231.65) and (173.38,232.81) .. (123.18,245.77) ;
\draw [shift={(121.67,246.17)}, rotate = 345.15999999999997] [color={rgb, 255:red, 0; green, 0; blue, 0 }  ][line width=0.75]    (10.93,-3.29) .. controls (6.95,-1.4) and (3.31,-0.3) .. (0,0) .. controls (3.31,0.3) and (6.95,1.4) .. (10.93,3.29)   ;

\draw   (192.33,308.67) -- (214,308.67) -- (214,330.33) -- (192.33,330.33) -- cycle ;
\draw   (351,308) -- (372.67,308) -- (372.67,329.67) -- (351,329.67) -- cycle ;
\draw   (111,319.67) .. controls (111,313.22) and (116.22,308) .. (122.67,308) .. controls (129.11,308) and (134.33,313.22) .. (134.33,319.67) .. controls (134.33,326.11) and (129.11,331.33) .. (122.67,331.33) .. controls (116.22,331.33) and (111,326.11) .. (111,319.67) -- cycle ;
\draw   (191.67,358.67) -- (213.33,358.67) -- (213.33,380.33) -- (191.67,380.33) -- cycle ;
\draw   (111.67,358.67) -- (133.33,358.67) -- (133.33,380.33) -- (111.67,380.33) -- cycle ;
\draw   (269,369) .. controls (269,362.56) and (274.22,357.33) .. (280.67,357.33) .. controls (287.11,357.33) and (292.33,362.56) .. (292.33,369) .. controls (292.33,375.44) and (287.11,380.67) .. (280.67,380.67) .. controls (274.22,380.67) and (269,375.44) .. (269,369) -- cycle ;
\draw   (351,357.33) -- (372.67,357.33) -- (372.67,379) -- (351,379) -- cycle ;
\draw  [dash pattern={on 4.5pt off 4.5pt}]  (359.67,246.83) .. controls (334.19,234.42) and (309.35,235.45) .. (284.52,248.05) ;
\draw [shift={(283,248.83)}, rotate = 332.24] [color={rgb, 255:red, 0; green, 0; blue, 0 }  ][line width=0.75]    (10.93,-3.29) .. controls (6.95,-1.4) and (3.31,-0.3) .. (0,0) .. controls (3.31,0.3) and (6.95,1.4) .. (10.93,3.29)   ;

\draw (112,16) node [scale=0.8] [align=left] {$\displaystyle P_{3}$};
\draw (221.33,28) node [scale=0.8] [align=left] {$\displaystyle P_{2} \ P_{4}$};
\draw (373.33,28.33) node [scale=0.8] [align=left] {$\displaystyle P_{1}$};
\draw (66.67,150.33) node [scale=0.8] [align=left] {object};
\draw (120,150) node [scale=0.8] [align=left] {1};
\draw (202,150) node [scale=0.8] [align=left] {2};
\draw (280,148.67) node [scale=0.8] [align=left] {3};
\draw (360.67,149.33) node [scale=0.8] [align=left] {4};
\draw (66.67,177) node [scale=0.8] [align=left] {agent};
\draw (120,176.67) node [scale=0.8] [align=left] {1};
\draw (202,176.67) node [scale=0.8] [align=left] {2};
\draw (280,175.33) node [scale=0.8] [align=left] {3};
\draw (360.67,176) node [scale=0.8] [align=left] {4};
\draw (66.67,212.33) node [scale=0.8] [align=left] {step 1};
\draw (120,212) node [scale=0.8] [align=left] {1};
\draw (202,212) node [scale=0.8] [align=left] {2};
\draw (280,210.67) node [scale=0.8] [align=left] {3};
\draw (360.67,211.33) node [scale=0.8] [align=left] {4};
\draw (67.33,262.67) node [scale=0.8] [align=left] {step 2};
\draw (120.67,262.33) node [scale=0.8] [align=left] {1};
\draw (202.67,262.33) node [scale=0.8] [align=left] {2};
\draw (280.67,261) node [scale=0.8] [align=left] {3};
\draw (361.33,261.67) node [scale=0.8] [align=left] {4};
\draw (66.67,318.33) node [scale=0.8] [align=left] {step 3};
\draw (204,319) node [scale=0.8] [align=left] {2};
\draw (282,317.67) node [scale=0.8] [align=left] {4};
\draw (362.67,318.33) node [scale=0.8] [align=left] {1};
\draw (123.17,319.5) node [scale=0.8] [align=left] {3};
\draw (66,368.33) node [scale=0.8] [align=left] {step 4};
\draw (203.33,369) node [scale=0.8] [align=left] {2};
\draw (362,368.33) node [scale=0.8] [align=left] {1};
\draw (122.5,369.5) node [scale=0.8] [align=left] {3};
\draw (282,369) node [scale=0.8] [align=left] {4};

\end{tikzpicture}
\caption{The descending crawler.}
\label{CRexamplecont}
\end{figure}
\end{example} 

In the example, the assignments obtained by applying the two rules are the same.  In general, one would expect the two rules to select different allocations.  However, it turns out that these allocations are always the same:  

\begin{theorem}\label{symmetry}
The ascending crawler is the same rule as the descending crawler.  That is, $ACR = DCR$.  
\end{theorem}

Because the ascending crawler visits agents from left to right, one might have conjectured that this specific order confers a particular advantage to some agents based on the location of their endowments relative to the location of other agents' endowments.  Specifically, because agents whose endowments are on the left side of the order are visited first, one may suspect that these agents enjoy some advantage at the expense of the agents whose endowments are on the right side of the order.  However, Theorem \ref{symmetry} implies that the order in which agents are visited gives no benefit to any particular agent.     

Hereafter, for simplicity, we refer to these rules as the \textbf{crawler}, using the notation $\bm{CR}$.  Without loss of generality, for each endowment profile, when we operate the crawler, we apply the ascending procedure.

\subsection{The second equivalence result: The crawler is invariant to the order over the object set that preserves single-peakedness} 
There may be multiple orders on the object set with respect to which a preference profile is single-peaked. 
If a preference profile is single-peaked with respect to $\prec$, it is obviously single-peaked with respect to the reverse order.  Then Theorem \ref{symmetry} implies the following result:

\begin{cor}\label{cor_inv}
Let $\prec,\prec' \in \mathcal{L}$ be such that $\prec$ and $\prec'$ are reverse to each other. For each $(P,\omega) \in \mathcal{P}^N \times \mathcal{X}$ such that $P$ is single-peaked with respect to $\prec$ and $\prec'$, 
$$CR^\prec(P,\omega) = CR^{\prec'}(P,\omega).$$
\end{cor}

A preference profile may be single-peaked with respect to more than one order and its reverse, however. Here is an example of such a preference profile. 

\begin{example}\label{object_order}
\textit{A preference profile that is single-peaked with respect to more than one order and its reverse.} Let $N = \{1,2,3,4,5,6\}$.  Let $P \in \mathcal{P}^N$ be defined by 
\begin{align*}
P_{1,2,3}&:~\omega_1,~\omega_2,~\omega_3,~\omega_4,~\omega_5,~\omega_6 \\
P_{4,5,6}&:~\omega_2,~\omega_1,~\omega_5,~\omega_3,~\omega_6,~\omega_4.  
\end{align*}
There are four orders over the object set with respect to which the profile is single-peaked;
\begin{align*}
\begin{array}{lcl}
\omega_6 \prec \omega_5 \prec \omega_1 \prec \omega_2 \prec \omega_3 \prec \omega_4 
&
\text{and its reverse}
&
\omega_4\prec' \omega_3 \prec' \omega_2 \prec' \omega_1 \prec' \omega_5\prec' \omega_6 \\
\omega_6~\hat{\prec}~\omega_5~\hat{\prec}~\omega_2~\hat{\prec}~\omega_1~\hat{\prec}~\omega_3~\hat{\prec}\omega_4
&
\text{and its reverse}
&
\omega_4~\tilde{\prec}~\omega_3~\tilde{\prec}~\omega_1~\tilde{\prec}~\omega_2~\tilde{\prec}~\omega_5~\tilde{\prec}~\omega_6. 
\end{array}
\end{align*} 
\end{example} 

$ $ 

One may think that the crawler selects different allocations for the various orders with respect to which a preference profile is single-peaked. 
It turns out that the crawler is invariant with respect to the choice of such an order.
For each $\prec \in \mathcal{L}$, we denote by $CR^\prec$ the crawler with respect to order $\prec$. 
 
\begin{theorem}\label{order_inv}
The crawler is invariant with respect to the choice of an order on the object set that preserves single-peakedness of preference profiles.  That is, for each $(P,\omega) \in \mathcal{P}^N \times \mathcal{X}$ and each pair $\prec,\prec' \in \mathcal{L}$ such that $P$ is single-peaked with respect to $\prec$ and $\prec'$, 
\[
CR^\prec(P,\omega)~=~CR^{\prec'}(P,\omega).
\]  
\end{theorem}

Theorem \ref{order_inv} states that the crawler is robust not only to the choice of of an order on the object set and its reverse with respect to which single-peakedness holds but also to the choice of any order with respect to which a given preference profile is single-peaked. 
This result provides an additional fact that the crawler does not favor particular agents on the basis of the location of their endowments at the expense of the others.
\section{The third equivalence result: Equivalence between the crawler from random endowments and the random priority rule}\label{prob}

Instead of each agent being endowed with one object, agents may collectively own a set of objects and a rule then has to assign each agent one object.  We refer to this type of problems as ``object allocation problems".  
Although the crawler is defined for reallocation problems, it can help provide solutions to object allocation problems.  To explain how, we first define the concept of a ``sequential priority rule": To each order on the agent set is associated such a rule.  The first agent in the order receives his most preferred object, the second agent in the order receives his most preferred object among the remaining ones, and so on (the formal definition is given below).  None of these rules treats agents fairly.  However, their unfairness can be mitigated by considering lotteries.  Formally, a \textbf{lottery} is a probability distribution over allocations, $p~=~(p_1,\cdots,p_{n!})$, such that for each $k$, $p_k \geq 0$, and $\sum_k p_k = 1$.  
We denote the degenerate lottery that assigns probability 1 to allocation $x$ by $p^x$.  Let $\Delta(\mathcal{X})$ be the set of all lotteries. 
A \textbf{probabilistic rule} is a single-valued mapping which associates a lottery with each preference profile.  The average of all sequential priority rules is obtained by choosing the order on the agent set at random according to a uniform distribution and applying the induced sequential priority rule.  The rule so defined is called the ``random priority rule" \citep{as98}. 

We study the relationship between the random priority rule and the \textbf{crawler from random endowments}: For each problem, choose an endowment profile at random according to a uniform distribution, and apply the crawler to the induced object reallocation problem.  
The procedures underlying the crawler and the sequential priority rules are quite different.  Thus, one should be doubtful that there is any relationship between the random priority rule and the crawler from random endowments.  However, as we show, the probability distributions over allocations obtained by applying these two rules are the same (Theorem \ref{thm_random}). 

Let us now formally define the family of rules under discussion.  Here, a problem is simply defined as a preference profile.  However, the crawler selects an allocation on the basis of both a preference profile and an endowment profile, so that the family of rules that we are defining are parametrized by the endowment profile $\omega$.  We indicate this parametrization with the superscript $\omega$, and denote the rule associated with the parameter $\omega$ by $CR^\omega$.  

Given $f \in \mathcal{F}$, the \textbf{sequential priority rule induced by $\bm{f}$}, $SP^f:\mathcal{P}^N \rightarrow \mathcal{X}$, is defined by setting, for each $P \in \mathcal{P}^N$, 
\begin{align*}
SP^f_{f(1)}(P) &= X_{f(1)}(O), \\
SP^f_{f(2)}(P) &= X_{f(2)}\left(O ~\backslash~ \left\{SP^f_{f(1)}(P)\right\}\right), \\
 & \vdots \\
SP^f_{f(i)}(P) &= X_{f(i)}\left(O~\backslash~ \bigcup_{j = 1}^{i-1} \left\{SP^f_{f(j)}(P)\right\}\right), \\
 & \vdots \\
SP^f_{f(n)}(P) &= X_{f(n)} \left(O~\backslash~ \bigcup_{j = 1}^{n-1} \left\{SP^f_{f(j)}(P)\right\}\right). 
\end{align*}

The \textbf{random priority rule}, $RP: \mathcal{P}^N \rightarrow \Delta(\mathcal{X})$, is defined by setting, for each $P \in \mathcal{P}^N$, 
\[RP(P) = \sum_{f \in \mathcal{F}} \frac{1}{n!} p^{SP^f(P)}.\]
  
The \textbf{crawler from random endowments} $RCR: \mathcal{P}^N~\rightarrow~\Delta(\mathcal{X})$, is defined by setting, for each $P \in \mathcal{P}^N$,
\[RCR(P) = \sum_{\omega \in \mathcal{X}} \frac{1}{n!} p^{CR^\omega(P)}.\]

Here is our third equivalence result:

\begin{theorem} \label{thm_random}
The crawler from random endowments is the same probabilistic rule as the random priority rule. That is, $RCR = RP$.  
\end{theorem}

The proof involves constructing a mapping $g: \mathcal{X} \rightarrow \mathcal{F}$ by recursively finding pairs of agents such that for each pair, one ``envies" the other, and giving higher priority to the second agent than to the first agent. Then we show that (\rom{1}) the allocation selected by the crawler given an endowment profile is the same as the allocation selected by the sequential priority rule induced by the order on the agent set that is selected by applying the mapping to that endowment profile, and (\rom{2}) the mapping is one-to-one and onto.  

\subsection{Discussion}\label{dis}
Following the equivalence result between the core from random endowments and the random priority rule (\citealp{knuth96}; \citealp{as98}), several other equivalence results have been proved (\citealp{ps11}; \citealp{ls11}; \citealp{su05}; \citealp{ekici17}; \citealp{carroll14}; \citealp{bade19equiv}).  Each of these papers generalizes TTC in some fashion, and shows that the probabilistic rule associated with the generalized TTC is the same as the random priority rule (or a variant of it).  None of these rules are equivalent to the crawler.  Hence, our Theorem \ref{thm_random} cannot be deduced from any of these results.   

A large family of rules, which are \textit{efficient}, \textit{strategy-proof}, and \textit{non-bossy},\footnote{Formally, $\varphi$ is \textit{non-bossy} if for each $(P,\omega) \in \mathcal{P}^N \times \mathcal{X}$, each $i \in N$, and each $P'_i \in \mathcal{P}$, if $\varphi_i\left(\left(P'_i,P_{-i}\right),\omega\right)~=~\varphi_i\left(\left(P_i,P_{-i}\right),\omega\right)$, then $\varphi\left(\left(P'_i,P_{-i}\right),\omega\right) = \varphi\left(\left(P_i,P_{-i}\right),\omega\right)$.}  is defined by \cite{pu17}.
They show that on the domain of strict preferences, any rule satisfying the above properties is a member of this family. 
Rules in their family, the so-called ``trading cycles" rules, are parametrized by what they call a ``control-rights structure" (the formal definition is given in the appendix).  
The procedure underlying their definition is similar to TTC, but control-rights structures add flexibility to TTC.  Now, given a control-rights structure, the associated probabilistic rule is obtained by permuting the agent set at random according to a uniform distribution, and applying the induced rule.  

For each rule in the trading cycles family, the probabilistic rule associated with it is the same rule as the random priority rule \citep{bade19equiv}.  
This result holds for each preference profile, and in particular of course, it holds on the subdomain of problems with profiles of single-peaked preferences.  Thus, on this domain, her equivalence result remains true; yet, on this domain, there may be rules that are not trading cycles rules, but are still \textit{efficient}, \textit{strategy-proof}, and \textit{non-bossy}.  In fact, the crawler is an example if such a rule (Appendix~\ref{example_ap}). 
Therefore, our Theorem~\ref{thm_random} cannot be deduced from \cite{bade19equiv}'s previous result.    

\section{Conclusion}
We have shown three equivalence results pertaining to a new rule defined by \cite{bade19} for object reallocation problems when preferences are single-peaked.
Our first equivalence result states that this rule, which she calls the crawler and which we call the ascending crawler, is the same as the dual rule that she proposes, and which we call the descending crawler.  Thus, the order in which agents are visited does not confer any particular advantage to some agents based on the location of their endowments relative to the location of the other agents' endowments. 

Furthermore, single-peakedness of a preference profile may hold with respect to more than one order on the object set that differs from the reverse of another order. 
Our second equivalence result states that for each single-peaked preference profile, the crawler is invariant with respect to the choice of an order on the object set that preserves single-peakedness of the profile.  This result provides additional fact that the crawler does not favor particular agents on the basis of the location of their endowments at the expense of the others.   

Our third equivalence result concerns object allocation problems, and state that the probability distribution over allocations selected by the crawler from random endowments is the same as the probability distribution selected by the random priority rule.  This equivalence result provides another structural analysis of the crawler.

\begin{appendices}


\section{Proof of Theorem \ref{symmetry}}
We use two facts to prove Theorem \ref{symmetry}.
First, because both the ascending and the descending crawlers are \emph{individually rational}, 

\begin{fact}\label{claim1_theorem1}
For each $(P,\omega) \in \mathcal{P}^N \times \mathcal{X}$ and each $i \in N$, 
$$p(P_i) = \omega_i~\Longrightarrow~ACR_i(P,\omega) = \omega_i = DCR_i(P,\omega).$$
\end{fact}

$ $ 

Second for each problem and each pair of adjacent agents,  if each of them prefers the endowment of the other agent to his own,  the allocation selected by the ascending crawler remains the same even if they swap their endowments before the ascending crawler is applied \citep{tamura21}.  Also, an analogous argument holds for the descending crawler.\footnote{\cite{tamura21} calls the ascending crawler simply the crawler.}

Formally, for each $\omega \in \mathcal{X}$ and each pair $i,j \in N$, let $\omega^{i,j} \in \mathcal{X}$ be obtained from $\omega$ by swapping the endowments of agent $i$ and $j$. That is, $\omega^{i,j}_i = \omega_j$, $\omega^{i,j}_j = \omega_i$, and for each $k \in N \setminus \{i,j\}$, $\omega^{i,j}_k = \omega_k$. 

\begin{fact}[\citealp{tamura21}]\label{claim2_theorem1}
For each $(P,\omega) \in \mathcal{P}^N \times \mathcal{X}$ and each $i \in N$, 
$$\omega_{i+1}~P_i~\omega_i~\text{and}~\omega_i~P_{i+1}~\omega_{i+1}~\Longrightarrow~
~ACR(P,\omega^{i,i+1}) = ACR(P,\omega)~\text{and}~DCR(P,\omega^{i,i+1}) = DCR(P,\omega).$$
\end{fact}

$ $ 

\begin{proof}[Proof of Theorem \ref{symmetry}]
Let $(P,\omega) \in \mathcal{P}^N \times \mathcal{X}$.  The proof is by induction on $n$ (the number of agents). 

Suppose that $n = 1$.  Then $ACR(P,\omega) = \omega = DCR(P,\omega)$. 

Let $k \in \n$. Suppose that when $n < k$, 
$$ACR(P,\omega) = \omega = DCR(P,\omega).$$
We show that when $n = k$, 
\begin{equation}\label{eq:symmetry1}
ACR(P,\omega) = DCR(P,\omega).
\end{equation} 

Without loss of generality, suppose that agent 1 is the first agent who is assigned an object when the ascending crawler is applied to $(P,\omega)$. 
Suppose that $p(P_1) = \omega_1$. Then by Fact \ref{claim1_theorem1},
$$ACR_1(P,\omega) = DCR_1(P,\omega).$$
After agent 1 leaves with his assignment, there are $k-1$ remaining agents. By the induction hypothesis, \eqref{eq:symmetry1} holds. 
Suppose that $p(P_1) \neq \omega_1$. Note that for each $i \in N$ such that $\omega_i \prec \omega_1$, we have $\omega_i \prec p(P_i)$. We call agent $a$ the agent whose endowment is assigned to agent~1 when the ascending crawler is applied to $(P,\omega)$, i.e., $\omega_{a} = ACR_1(P,\omega)$. This means that $\omega_{a} = p(P_1)$.  Also, let $\{i_1,\hdots,i_k\} \subset N$ be the set of agents whose endowments are between $\omega_{a}$ and $\omega_1$ and $\omega_{i_k} \prec \hdots \prec \omega_{i_1}$. 
Starting with $\omega$, suppose that agents 1 and $i_1$ swap their endowments before either the ascending or the descending crawler is applied to $(P,\omega)$. By Fact \ref{claim2_theorem1},
$$ACR(P,\omega) = ACR(P,\omega^{1,i_1})~\text{and}~DCR(P,\omega) = DCR(P,\omega^{1,i_1}).$$
Let $\tilde{\omega} \equiv \omega^{1,i_1}$. 
Starting with $\tilde{\omega}$, suppose that agents 1 and $i_2$ swap their endowments before the ascending or the descending crawler is applied to $(P,\tilde{\omega})$. By Fact \ref{claim2_theorem1},
$$ACR(P,\omega) = ACR(P,\tilde{\omega}) = ACR(P,\tilde{\omega}^{1,i_2})~\text{and}~DCR(P,\omega) = DCR(P,\tilde{\omega}) = DCR(P,\tilde{\omega}^{1,i_2}).$$
Let $\hat{\omega} \in \mathcal{X}$ be such that $\hat{\omega}_1 = \omega_{a}$, $\hat{\omega}_{a} = \omega_{i_k}$, and for each $l \in \{1,\hdots,k\}$, $\hat{\omega}_{i_l} = \omega_{i_{l-1}}$ (where $\hat{\omega}_{i_1} = \omega_1$). 
By an analogous argument, 
$$ACR(P,\omega) = ACR(P,\hat{\omega})~\text{and}~DCR(P,\omega) = DCR(P,\hat{\omega}).$$
By Fact \ref{claim1_theorem1},
$$ACR_1(P,\omega) = ACR_1(P,\hat{\omega}) = \hat{\omega}_1 = DCR_1(P,\hat{\omega}) = DCR_1(P,\omega).$$
After agent 1 leaves with his assignment, there are $k-1$ remaining agents. By the induction hypothesis, \eqref{eq:symmetry1} holds. 
\end{proof}

\section{Proof of Theorem \ref{order_inv}}

Let $\prec,\prec' \in \mathcal{L}$ be such that $\prec'$ is not the reverse of $\prec$. Let $P \in \mathcal{P}^N$ be single-peaked with respect to both $\prec$ and $\prec'$. 
Let $\underline{o},\overline{o} \in O$ be such that there is pair $i,j \in N$ such that $p(P_i) = \underline{o}$ and $p(P_j) = \overline{o}$, and for each $i \in N$, $\underline{o} \precsim p(P_i) \precsim \overline{o}$. By Corollary \ref{cor_inv}, without loss of generality, suppose that $\underline{o} \precsim' \overline{o}$. 

\begin{claim}\label{order_inv_claim1}
For each pair $o,o' \in O$, 
$$\underline{o} \prec o \prec o' \prec \overline{o}~\iff~\underline{o} \prec' o \prec' o' \prec' \overline{o}.$$
\end{claim}

\begin{proof}[Proof of Claim \ref{order_inv_claim1}]
Let $o \in O$ be such that $\underline{o} \prec o \prec \overline{o}$. 
Let $i,j \in N$ be such that $p(P_i) = \underline{o}$ and $p(P_j) = \overline{o}$. 
Because $P_i$ and $P_j$ are single-peaked with respect to $\prec$, 
\begin{equation}\label{eq:theorem21}
\underline{o}~P_i~o~P_i~\overline{o}~\text{and}~\overline{o}~P_j~o~P_j~\underline{o}.
\end{equation}
Agent $i$'s preferences imply that one of the following conditions holds:
$$\text{(i)}~\underline{o} \prec' o \prec' \overline{o},~\text{(ii)}~\overline{o} \prec' o \prec' \underline{o},~\text{(iii)}~o \prec' \underline{o} \prec' \overline{o},~\text{or (iv)}~\overline{o} \prec' \underline{o} \prec' o.$$
Conditions (ii) and (iv) contradict the assumption that $\underline{o} \precsim' \overline{o}$. Condition (iii) contradicts the hypothesis that $P_j$ is single-peaked with respect to $\prec'$. Therefore, Condition (i) holds. 

Now let $o,o' \in O$ be such that $\underline{o} \prec o \prec o' \prec \overline{o}$. 
By the above argument, $\underline{o} \prec' o,o' \prec' \overline{o}$. 
Note that $\underline{o}~P_i~o~P_i~o'~P_i~\overline{o}$. Because $P_i$ is single-peaked with respect to $\prec'$, $o \prec' o'$. 
\end{proof}

$ $ 

\begin{proof}[Proof of Theorem \ref{order_inv}]
Let $\prec,\prec' \in \mathcal{L}$ be such that $\prec \neq \prec'$. 
Let $(P,\omega) \in \mathcal{P}^N \times \mathcal{X}$ be such that $P$ is single-peaked with respect to both $\prec$ and $\prec'$.  The proof is by induction on $n$ (the number of agents). 

Suppose that $n =1$. Then $CR^\prec(P,\omega) = CR^{\prec'}(P,\omega)$. 
Let $k \in \n$. Suppose that when $n < k$, we have $CR^\prec(P,\omega) = CR^{\prec'}(P,\omega)$. 

We show that when $n = k$, 
\begin{equation}\label{eq:inv1}
CR^\prec(P,\omega) = CR^{\prec'}(P,\omega).
\end{equation} 
Suppose that for each $o \in O$, $\underline{o} \precsim o \precsim \overline{o}$. This implies $\prec'$ is the reverse of $\prec$. By Theorem \ref{symmetry}, \eqref{eq:inv1} holds. 
Suppose that there is $o \in O$ such that either $o \prec \underline{o}$ or $\overline{o} \prec o$.  By Claim \ref{order_inv_claim1}, for each pair $o,o' \in O$, $\underline{o} \prec o \prec o' \prec \overline{o}$ if and only if $\underline{o} \prec' o \prec' o' \prec' \overline{o}$. 
Without loss of generality, suppose that agent 1 is the first agent who is assigned an object when $CR^\prec$ is applied to $(P,\omega)$. 
Suppose that $p(P_1) = \omega_1$. By Fact \ref{claim1_theorem1}, 
$$CR^\prec_1(P,\omega) = \omega_1 = CR^{\prec'}_1(P,\omega).$$
After agent 1 leaves with his assignment, there are $k-1$ remaining agents. By the induction hypothesis, \eqref{eq:inv1} holds.
Suppose that $p(P_1) \neq \omega_1$. 
Note that for the agents whose endowments are between $\underline{o}$ and $\overline{o}$, agent 1 is the first agent who is assigned an object when $CR^{\prec'}$ is applied to $(P,\omega)$. 
Hence, for each $i \in N$ such that $\underline{o} \precsim' \omega_i \prec \omega_1$, we have $\omega_i \prec' p(P_i)$. 
We call agent $a \in N$ the agent whose endowment is assigned to agent 1 when $CR^{\prec'}$ is applied to $(P,\omega)$, i.e., $\omega_{a} = CR^{\prec'}_1(P,\omega)$. Note that $\omega_{a} = p(P_1)$.  Also, let $\{i_1,\hdots,i_k\} \subset N$ be the set of agents whose endowments are between $\omega_{a}$ and $\omega_1$ and $\omega_{i_k} \prec \hdots \prec \omega_{i_1}$. Hence, 
$$\omega_{a} \prec' \omega_{i_k} \prec' \hdots \prec' \omega_{i_1} \prec' \omega_1.$$
Let $\hat{\omega} \in \mathcal{X}$ be such that $\hat{\omega}_1 = \omega_{a}$, $\hat{\omega}_{a} = \omega_{i_k}$, and for each $l \in \{1,\hdots,k\}$, $\hat{\omega}_{i_l} = \omega_{i_{l-1}}$ (where $\hat{\omega}_{i_1} = \omega_1$). 
By applying Fact \ref{claim2_theorem1},
$$CR^{\prec'}(P,\hat{\omega}) = CR^{\prec'}(P,\omega).$$ 
By Fact \ref{claim1_theorem1}, 
$$CR^{\prec'}_1(P,\omega) = p(P_1) = CR^\prec_1(P,\omega).$$
After agent 1 leaves with his endowment, there are $k-1$ remaining agents. By the induction hypothesis, \eqref{eq:inv1} holds. 
\end{proof}

\section{Proof of Theorem \ref{thm_random}}

Let $P \in \mathcal{P}^N$. 
The following lemma shows that for each allocation selected by the sequential priority rule induced by a given order on the agent set, there is an endowment profile for which the crawler selects the same allocation at the endowment profile.  Conversely, for each allocation selected by the crawler at a given endowment profile, there is an order on the agent set such that the sequential priority rule induced by the order selects the same allocation.    

\begin{lemma} \label{set-equality}
\[\left\{x \in \mathcal{X}:~\text{there is}~f \in \mathcal{F}~\text{such that}~SP^f(P) = x\right\} = \left\{x \in \mathcal{X}:~\text{there is}~\omega \in \mathcal{X}~\text{such that}~CR^\omega(P) = x\right\}.\]
\end{lemma}

Given a preference profile, the set of allocations selected by TTC at the various endowment profiles is the same as the set of allocations selected by the sequential priority rules induced by the various orders on the agent set (\citealp{knuth96}; \citealp{as98}).\footnote{\cite{as98} also show that for each allocation in the set, the frequency of the selection of that allocation by TTC is the same as that by the sequential priority rules.}  We show that the set of allocations selected by the crawler at the various endowment profiles is the same as the set of allocations selected by TTC at the various endowment profiles.  Combining these results, we derive Lemma \ref{set-equality}.  

\begin{lemma*}[\citealp{knuth96}; \citealp{as98}]
\begin{equation}\label{res_as98}
\left\{x \in \mathcal{X}:~\text{there is}~f \in \mathcal{F}~\text{such that}~SP^f(P) = x\right\} = \left\{x \in \mathcal{X}:~\text{there is}~\omega \in \mathcal{X}~\text{such that}~TTC^\omega(P) = x\right\}.
\end{equation}
\end{lemma*}

\begin{proof}[Proof of Lemma \ref{set-equality}]
We show that for each $\omega \in \mathcal{X}$, there is $\omega' \in \mathcal{X}$ such that \[TTC^\omega(P) = CR^{\omega'}(P).\]  
This implies that 
\[\left\{x \in \mathcal{X}:~\text{there is}~\omega \in \mathcal{X}~\text{such that}~TTC^\omega(P) = x\right\} \subseteq \left\{x \in \mathcal{X}:~\text{there is}~\omega \in \mathcal{X}~\text{such that}~CR^\omega(P) = x\right\}.\]
The proof for the other direction is analogous.  Hence, we only show one direction. 
 
Let $\omega \in \mathcal{X}$.   
Because the crawler is \textit{individually rational}, for each $i \in N$, 
\[CR_i^{TTC^\omega(P)}(P)~R_i~TTC_i^\omega(P).\] 
Because TTC is \textit{efficient}, there is no $i \in N$ such that 
\[CR_i^{TTC^\omega(P)}(P)~P_i~TTC_i^\omega(P).\]  
This implies that 
\[CR^{TTC^\omega(P)}(P) = TTC^\omega(P).\]

Together with \eqref{res_as98}, we have 
\[\left\{x \in \mathcal{X}:~\text{there is}~f \in \mathcal{F}~\text{such that}~SP^f(P) = x\right\} = \left\{x \in \mathcal{X}:~\text{there is}~\omega \in \mathcal{X}~\text{such that}~CR^\omega(P) = x\right\}.\]
\end{proof}

\bigskip

We now construct a mapping $g$ from $\mathcal{X}$ into $\mathcal{F}$ and show that (\rom{1}) the crawler induced by an endowment profile and the sequential priority rule induced by the order on the agent set given by the mapping $g$ select the same allocation, and (\rom{2}) the mapping is one-to-one and onto. 

Mapping $g$ is constructed by the procedure described below.  There are two phases.  In Phase~1, we construct ``envy graphs", and in Phase 2, we derive an order on the agent set. 
 
Phase 1 proceeds as follows:   
At Round 1, 
(1) for each agent $i \in N$, we identify $j~\in~N$ who receives agent $i$'s most preferred object.  If $j \neq i$, we say that \textbf{agent $\bm{i}$ envies agent~$\bm{j}$}. 
(2)~For each agent, if he is envied by someone and he envies someone else, we connect these envy relations.  For each sequence of agents such that each agent in the sequence, except for the first agent, envies the agent on his immediate left and no agent envies the last agent in the sequence, we connect all of these envy relations.  Because the crawler is \textit{efficient}, the first and the last agents cannot be the same, i.e., there is no cycle of envy.  We call the resulting connected envy relations a \textbf{maximal envy chain}.    
(3) For each pair of maximal envy chains formed at Round 1, if there is an agent who belongs to both, we ``attach" them at this agent.  The result is a subgraph of the envy graph. 
(4) We remove any agent who does not envy anyone from the preference profile.  Also, we update the preferences of the remaining agents as follows: for each agent who envies someone, we restrict his preferences to the objects that are not his most preferred object.  
At each Round $r \geq 2$, we repeat this procedure among the remaining agents.  Moreover, at Step 3, in addition to each pair of maximal envy chains formed at Round $r$, for each pair of maximal envy chains such that one forms at that round and the other forms at an earlier round, if there is an agent who belongs to both chains, we attach these chains at this agent. 
Phase 1 terminates when all agents are removed. The resulting graph is the envy graph. 

Phase 2 proceeds as follows:  
(1) initially, the agents are ordered in such a way that for each pair $i,j \in N$, agent $i$ comes earlier than agent $j$ if and only if $\omega_i \prec \omega_j$.  
(2) For the agents who form a component of the envy graph, we permute their positions in the order based on the following two criteria. (i) Consider a set of agents who form a maximal envy chain. For each pair of agents in the chain, if one envies the other, the second agent comes earlier than the first agent in the updated order. (ii) For each pair of agents who do not belong to the same maximal envy chain, their relative positions in the order are determined according to the positions of their endowments and their assignments as described below. 

\bigskip 

Here is the formal definition.  
For each agent, we define two sets of objects, a \textbf{set of eliminated objects} and a \textbf{set of residual objects}, that are updated at each round.
Recall that for each $P_i \in \mathcal{P}$ and each $O' \subseteq O$, we denote by $X_i(O')$ the most preferred object of agent $i$ in $O'$ at $P_i$.

Let $\omega \in \mathcal{X}$.  

\noindent\textbf{Phase 1}: \textbf{Construct the envy graph.}  

\noindent\textbf{Round 0}:
For each $i \in N$, let 
\[
O^0_i = O~\text{and}~E^0_i = \emptyset. 
\]

\noindent\textbf{Round $\bm{r \geq 1}$}:

\noindent\textbf{Step 1}: \textbf{Identify each pair $\bm{i,j \in N}$ such that agent $\bm{i}$ envies agent $\bm{j}$ for $\bm{X(O_i^{r-1})}$.}

For each $i \in N$ such that $O^{r-1}_i \neq \emptyset$, identify $j \in N$ such that $CR^\omega_j(P)~=~X_i(O^{r-1}_i)$.  If $j \neq i$, write $i \rightarrow j$.  Also, let $E^r_i~=~\{X_i(O^{r-1}_i)\}$.  If $j = i$, let $E^r_i = O^{r-1}_i$. 
Notice that each agent envies at most one agent at each round. 

\noindent\textbf{Step 2}: \textbf{Identify all maximal envy chains.}\footnote{For each sequence of agents $\{i_1,i_2,\hdots,i_k\} \subseteq N$ for some $k \in \{1,\hdots,n\}$ such that for each $k'~\in~\{2,\hdots,k\}$, $i_{k'} \rightarrow i_{k'-1}$, agent $i_1$ does not point to any agent, and no agent points to $i_k$, we concatenate these relations, i.e., 
\[
i_k~\rightarrow~i_{k-1}~\rightarrow \hdots \rightarrow i_2~\rightarrow~i_1. 
\]
}

\noindent\textbf{Step 3}: \textbf{Construct subgraphs of the envy graph.} 


For each pair of maximal envy chains such that one forms at Round $r$ and the other forms at Round $r' \leq r$, if there is an agent who belongs to both, attach them at this agent. 
There are three possible configurations described below.  Each component in a subgraph of the envy graph is either a maximal envy chain or a combination of these three configurations. 

Let $I,J \subseteq N$ be such that the agents in $I$ form a maximal envy chain at Round $r'$ and the agents in $J$ form a maximal envy chain at Round $r$.  Let $k,k'~\in~\{1,\hdots,n-1\}$.

\noindent\textbf{Configuration 1}:
There is $\{i_0,\hdots,i_{k}\} \subseteq I$ and $\{j_0,\hdots,j_{k'}\} \subseteq J$, where $i_{0}~=~ j_{0}~\equiv~l$, such~that 
\begin{align*}
i_k~\rightarrow \hdots \rightarrow~i_1~\rightarrow~l~\text{and}~
j_{k'}~\rightarrow \hdots \rightarrow~j_1~\rightarrow~l. 
\end{align*}
Connect them as described in Figure \ref{Configuration 1}.

\noindent\textbf{Configuration 2}: 
There is $\{i_1,\hdots,i_{k+1}\} \subseteq I$ and $\{j_1,\hdots,j_{k'+1}\} \subseteq J$, where $i_{k+1}~=~ j_{k'+1}~\equiv~l$, such~that 
\begin{align*}
l~\rightarrow~i_k~\rightarrow \hdots \rightarrow~i_1~\text{and}~
l~\rightarrow~j_{k'}~\rightarrow \hdots \rightarrow~j_{1}. 
\end{align*}
Connect them as described in Figure \ref{Configuration 2}.

\noindent\textbf{Configuration 3}: 
There is $\{i_1,\hdots,i_{k}\} \subseteq I$ and $\{j_1,j_2\} \subseteq J$, where $i_1 = j_1$ and $i_k = j_2$, such~that
\begin{align*}
i_k~\rightarrow \hdots \rightarrow~i_1~\text{and}~j_2~\rightarrow~j_1.
\end{align*}
Connect them as described in Figure \ref{Configuration 3}. 

\begin{figure}[H]
\centering
\begin{subfigure}[b]{0.3\textwidth}
\centering
\tikzset{every picture/.style={line width=0.75pt}} 

\begin{tikzpicture}[x=0.75pt,y=0.75pt,yscale=-1,xscale=1]

\draw    (205,128) -- (225,138) ;
\draw [shift={(225,138)}, rotate = 207] [color={rgb, 255:red, 0; green, 0; blue, 0 }  ][line width=0.75]    (6.56,-1.97) .. controls (4.17,-0.84) and (1.99,-0.18) .. (0,0) .. controls (1.99,0.18) and (4.17,0.84) .. (6.56,1.97)   ;
\draw    (207,168) -- (225,160) ;
\draw [shift={(225,160)}, rotate = 515] [color={rgb, 255:red, 0; green, 0; blue, 0 }  ][line width=0.75]    (6.56,-1.97) .. controls (4.17,-0.84) and (1.99,-0.18) .. (0,0) .. controls (1.99,0.18) and (4.17,0.84) .. (6.56,1.97)   ;

\draw (120,120) node [anchor=north west][inner sep=0.75pt]  [font=\footnotesize]  {$i_{k} \rightarrow \cdots \rightarrow i_{1}$};
\draw (118,162) node [anchor=north west][inner sep=0.75pt]  [font=\footnotesize]  {$j_{k'} \rightarrow \cdots \rightarrow j_{1}$};
\draw (230,141) node [anchor=north west][inner sep=0.75pt]  [font=\footnotesize]  {$l$};

\end{tikzpicture}
\caption{Configuration 1}
\label{Configuration 1}
\end{subfigure}
\begin{subfigure}[b]{0.3\textwidth}
\centering
\tikzset{every picture/.style={line width=0.75pt}} 

\begin{tikzpicture}[x=0.75pt,y=0.75pt,yscale=-1,xscale=1]

\draw    (96.66,157.97) -- (112.53,169.31) ;
\draw [shift={(114.16,170.47)}, rotate = 215.54] [color={rgb, 255:red, 0; green, 0; blue, 0 }  ][line width=0.75]    (6.56,-1.97) .. controls (4.17,-0.84) and (1.99,-0.18) .. (0,0) .. controls (1.99,0.18) and (4.17,0.84) .. (6.56,1.97)   ;
\draw    (97.66,144.97) -- (114.97,134.04) ;
\draw [shift={(116.66,132.97)}, rotate = 507.72] [color={rgb, 255:red, 0; green, 0; blue, 0 }  ][line width=0.75]    (6.56,-1.97) .. controls (4.17,-0.84) and (1.99,-0.18) .. (0,0) .. controls (1.99,0.18) and (4.17,0.84) .. (6.56,1.97)   ;

\draw (122,123) node [anchor=north west][inner sep=0.75pt]  [font=\footnotesize]  {$i_{k} \ \rightarrow \ \cdots \ \rightarrow \ i_{1}$};
\draw (122,165) node [anchor=north west][inner sep=0.75pt]  [font=\footnotesize]  {$j_{k'} \ \rightarrow \ \cdots \ \rightarrow \ j_{1}$};
\draw (85,144) node [anchor=north west][inner sep=0.75pt]  [font=\footnotesize]  {$l$};

\end{tikzpicture}
\caption{Configuration 2}
\label{Configuration 2}
\end{subfigure}
\begin{subfigure}[b]{0.3\textwidth}
\centering
\tikzset{every picture/.style={line width=0.75pt}} 

\begin{tikzpicture}[x=0.75pt,y=0.75pt,yscale=-1,xscale=1]

\draw    (170,173) .. controls (185,154) and (235,154) .. (245,173) ;
\draw [shift={(247,175)}, rotate = 235] [color={rgb, 255:red, 0; green, 0; blue, 0 }  ][line width=0.75]    (6.56,-1.97) .. controls (4.17,-0.84) and (1.99,-0.18) .. (0,0) .. controls (1.99,0.18) and (4.17,0.84) .. (6.56,1.97)   ;

\draw (164,180.9) node [anchor=north west][inner sep=0.75pt]  [font=\footnotesize]  {$i_{k} \ \rightarrow \ \cdots \ \rightarrow \ i_{1}$};

\end{tikzpicture}
\caption{Configuration 3}
\label{Configuration 3}
\end{subfigure}
\caption{\small{Possible configurations of a component in a subgraph of the envy graph.}}
\end{figure}
Because the crawler is \textit{efficient}, no cycle forms.

\begin{figure}[H]
\centering
\begin{subfigure}{0.49\textwidth}
\centering
\tikzset{every picture/.style={line width=0.75pt}} 

\begin{tikzpicture}[x=0.75pt,y=0.75pt,yscale=-1,xscale=1]

\draw    (272.97,141.4) .. controls (287.55,115.39) and (319.21,103.36) .. (354.56,110.65) ;
\draw [shift={(356.17,111)}, rotate = 192.53] [color={rgb, 255:red, 0; green, 0; blue, 0 }  ][line width=0.75]    (10.93,-3.29) .. controls (6.95,-1.4) and (3.31,-0.3) .. (0,0) .. controls (3.31,0.3) and (6.95,1.4) .. (10.93,3.29)   ;

\draw (147.5,122.18) node [anchor=north west][inner sep=0.75pt]  [rotate=-25]  {$i_{2} \ \rightarrow \ i_{1} \ \rightarrow $};
\draw (136.71,195.31) node [anchor=north west][inner sep=0.75pt]  [rotate=-340]  {$j_{2} \ \rightarrow \ j_{1} \ \rightarrow $};
\draw (226,166.4) node [anchor=north west][inner sep=0.75pt]    {$l$};
\draw (232.56,161.57) node [anchor=north west][inner sep=0.75pt]  [rotate=-340]  {$\rightarrow \ m_{3} \ \rightarrow \ m_{2} \ \rightarrow \ m_{1}$};
\draw (238.75,168.96) node [anchor=north west][inner sep=0.75pt]  [rotate=-20]  {$\rightarrow \ k$};

\end{tikzpicture}
\caption{}
\label{graph_ex1}
\end{subfigure}
\begin{subfigure}{0.49\textwidth}
\centering
\tikzset{every picture/.style={line width=0.75pt}} 

\begin{tikzpicture}[x=0.75pt,y=0.75pt,yscale=-1,xscale=1]

\draw (154.8,158) node [anchor=north west][inner sep=0.75pt]    {$i$};
\draw (192,125) node [anchor=north west][inner sep=0.75pt]    {$j_{2} \ \rightarrow \ j_{1} \ $};
\draw (194,180) node [anchor=north west][inner sep=0.75pt]    {$k_{2} \ \rightarrow \ k_{1}$};
\draw (281.9,156.8) node [anchor=north west][inner sep=0.75pt]    {$l$};
\draw (165,146) node [anchor=north west][inner sep=0.75pt]  [rotate=-330]  {$\rightarrow $};
\draw (172,170) node [anchor=north west][inner sep=0.75pt]  [rotate=-30]  {$\rightarrow $};
\draw (264,138) node [anchor=north west][inner sep=0.75pt]  [rotate=-30]  {$\rightarrow $};
\draw (260,183) node [anchor=north west][inner sep=0.75pt]  [rotate=-330]  {$\rightarrow $};

\end{tikzpicture}
\caption{}
\label{graph_ex2}
\end{subfigure}
\caption{\small{Two examples of a subgraph of the envy graph.  Figure \ref{graph_ex1} is a combination of Configurations 1, 2, and 3.  Figure \ref{graph_ex2} is a combination of Configurations 1 and 2.}}
\end{figure}

\noindent\textbf{Step 4}: \textbf{Update each agent's set of residual objects.}

For each $i \in N$, let $O_i^r = O^{r-1}_i \backslash E_i^r$.  

\bigskip

Phase 1 terminates at the round when for each $i \in N$, $O^r_i = \emptyset$.
Let $S \in \{1,\hdots,n\}$.
Let $C = \{C_1,\hdots,C_S\}$ be the partition of $N$ such that for each $s \in \{1,\hdots,S\}$, the agents in $C_s$ form a component in the envy graph.\footnote{A component may contain only one agent.}

\bigskip

\noindent
\textbf{Phase 2}: \textbf{Derive an order on the agent set.}

Let $g^0(\omega)~\in~\mathcal{F}$ be defined by setting, for each pair $i,j \in N$, 
\[
g^0_i(\omega) < g^0_j(\omega)~\iff~\omega_i \prec \omega_j.
\]  

Let $g(\omega) \in \mathcal{F}$ be defined by setting for each $s \in \{1,\hdots,S\}$,
\[
\bigcup_{i \in C_s} g_i(\omega) = \bigcup_{i \in C_s} g^0_i(\omega). 
\]

Let $s \in \{1,\hdots,S\}$, and let $k,k' \in \{1,\hdots,n\}$.  
Let $\{i_1,\hdots,i_k,j_1,\hdots,j_{k'},l\} \subseteq N$.  

\noindent
\textbf{Case 1}: $C_s = \{i_1,\hdots,i_k\}$ is such that 
\[
i_k~\rightarrow~i_{k-1}~\rightarrow \hdots \rightarrow~i_{2}~\rightarrow~i_1.
\]
Let $g(\omega) \in \mathcal{F}$ be such that for each pair $m,m' \in \{1,\hdots,k\}$,
\[
g_{i_m}(\omega) < g_{i_{m'}}(\omega)~\iff~m < m'. 
\]

\noindent
\textbf{Case 2}: The agents in $C_s$ form a component in the envy graph that is as simple as one of the three configurations described above.    

\noindent\textbf{Configuration 1}: $C_s = \{i_1,\hdots,i_k,j_1,\hdots,j_{k'},l\}$ is such that the agents in $C_s$ form a component in the envy graph as described in Figure \ref{Configuration 1}. 
Let $g(\omega) \in \mathcal{F}$ be such that 
\begin{enumerate}
\item for each $i \in C_s$, $g_l(\omega) \leq g_i(\omega)$;
\item for each pair $m,m' \in \{1,\hdots,k\}$, $g_{i_m}(\omega) < g_{i_{m'}}(\omega)~\iff~m < m'$; and
\item for each pair $m,m' \in \{1,\hdots,k'\}$, $g_{j_m}(\omega) < g_{j_{m'}}(\omega)~\iff~m < m'$.
\end{enumerate}
First, the order between $i_1$ and $j_1$ is determined as follows.
Let $g(\omega) \in \mathcal{F}$ be such that if $\omega_{i_1} \prec \omega_{j_1}$ and there is no $q' \in \{2,\hdots,k'\}$ such that $\omega_{i_1} \prec \omega_{j_{q'}} \precsim CR^\omega_{i_1}(P),CR^\omega_{j_1}(P) \precsim \omega_{j_1}$, then $g_{i_1}(\omega)~<~g_{j_1}(\omega)$; otherwise, $g_{i_1}(\omega) > g_{j_1}(\omega)$.  
Without loss of generality, suppose that $g_{i_1}(\omega) < g_{j_1}(\omega)$.  Next find a pair of agents $i,j \in \{i_1,\hdots,i_k\}$ such that agent $j_1$ is located between agents $i$ and $j$ by following the procedure below. Then apply an analogous procedure to agent $j_2$, and so on. 
\begin{enumerate}
\item[]\textbf{Step 1}: Identify the smallest $p \in \{2,\hdots,k\}$ such that either (i) $\omega_{j_1} \prec \omega_{i_p}$, or (ii)~$\omega_{i_p}~\prec~\omega_{j_1}$ and there is $q' \in \{2,\hdots,k'\}$ such that $\omega_{i_p} \prec \omega_{j_{q'}} \precsim CR^\omega_{j_1}(P) \precsim \omega_{j_1}$. 
Let $p^1~=~p$ if such a value exists; otherwise let $p^1 = k+1$.  
Let $g(\omega) \in \mathcal{F}$ be such that $g_{i_{p^1-1}}(\omega)~<~g_{j_1}(\omega)~<~g_{i_{p^1}}(\omega)$.  If $p^1 = k+1$, the order on $C_s$ is determined.  

\item[]\textbf{Step $\bm{q \geq 2}$}: Identify the smallest $p \in \{p^{q-1},\hdots,k\}$ such that either (i) $\omega_{j_q} \prec \omega_{i_p}$, or (ii)~$\omega_{i_p}~\prec~\omega_{j_q}$ and there is $q' \in \{q+1,\hdots,k'\}$ such that $\omega_{i_p}~\prec~\omega_{j_{q'}}~\precsim CR^\omega_{j_q}(P) \precsim \omega_{j_q}$. 
Let $p^q = p$ if such a value exists; otherwise let $p^q = k+1$.  
Let $g(\omega) \in \mathcal{F}$ be such that $g_{i_{p^q-1}}(\omega) < g_{j_1}(\omega) < g_{i_{p^q}}(\omega)$.  If $p^q = k+1$, the order on $C_s$ is determined. 
\end{enumerate} 

\noindent\textbf{Configuration 2}: $C_s = \{i_1,\hdots,i_k,j_1,\hdots,j_{k'},l\}$ is such that the agents in $C_s$ form a component in the envy graph as described in Figure \ref{Configuration 2}. 
Let $g(\omega) \in \mathcal{F}$ be such that 
\begin{enumerate}
\item for each $i \in C_s$, $g_l(\omega) \geq g_i(\omega)$;
\item for each pair $m,m' \in \{1,\hdots,k\}$, $g_{i_m}(\omega) < g_{i_{m'}}(\omega)~\iff~m < m'$; and
\item for each pair $m,m' \in \{1,\hdots,k'\}$, $g_{j_m}(\omega) < g_{j_{m'}}(\omega)~\iff~m < m'$.
\end{enumerate}
First, the order between $i_1$ and $j_1$ is determined as follows.
Let $g(\omega) \in \mathcal{F}$ be such that if $\omega_{i_1} \prec \omega_{j_1}$ and there is no $q' \in \{2,\hdots,k'\}$ such that $\omega_{i_1} \prec \omega_{j_{q'}} \precsim CR^\omega_{i_1}(P),CR^\omega_{j_1}(P) \precsim \omega_{j_1}$, then $g_{i_1}(\omega)~<~g_{j_1}(\omega)$; otherwise, $g_{i_1}(\omega) > g_{j_1}(\omega)$.  
Without loss of generality, suppose that $g_{i_1}(\omega) < g_{j_1}(\omega)$.  Next find a pair of agents $i,j \in \{i_1,\hdots,i_k\}$ such that agent $j_1$ is located between agents $i$ and $j$ by following the procedure below. Then apply an analogous procedure to agent $j_2$, and so on. 
\begin{enumerate}
\item[]\textbf{Step 1}: Identify the smallest $p \in \{2,\hdots,k\}$ such that either (i) $\omega_{j_1} \prec \omega_{i_p}$, or (ii)~$\omega_{i_p}~\prec~\omega_{j_1}$ and there is $q' \in \{2,\hdots,k'\}$ such that $\omega_{i_p} \prec \omega_{j_{q'}} \precsim CR^\omega_{j_1}(P) \precsim \omega_{j_1}$. 
Let $p^1~=~p$ if such a value exists, otherwise let $p^1 = k+1$.  
Let $g(\omega) \in \mathcal{F}$ be such that $g_{i_{p^1-1}}(\omega)~<~g_{j_1}(\omega)~<~g_{i_{p^1}}(\omega)$.  If $p^1 = k+1$, the order on $C_s$ is determined.  

\item[]\textbf{Step $\bm{q \geq 2}$}: Identify the smallest $p \in \{p^{q-1},\hdots,k\}$ such that either (i) $\omega_{j_q} \prec \omega_{i_p}$, or (ii)~$\omega_{i_p}~\prec~\omega_{j_q}$ and there is $q' \in \{q+1,\hdots,k'\}$ such that $\omega_{i_p}~\prec~\omega_{j_{q'}}~\precsim CR^\omega_{j_q}(P) \precsim \omega_{j_q}$. 
Let $p^q = p$ if such a value exists; otherwise let $p^q = k+1$.  
Let $g(\omega) \in \mathcal{F}$ be such that $g_{i_{p^q-1}}(\omega) < g_{j_1}(\omega) < g_{i_{p^q}}(\omega)$.  If $p^q = k+1$, the order on $C_s$ is determined. 
\end{enumerate} 

\noindent\textbf{Configuration 3}: $C_s = \{i_1,\hdots,i_k\}$ is such that the agents in $C_s$ form a component in the envy graph as described in Figure \ref{Configuration 3}.
Let $g(\omega) \in \mathcal{F}$ be such that for each pair $m,m' \in \{1,\hdots,k\}$,
\[
g_{i_m}(\omega) < g_{i_{m'}}(\omega)~\iff~m < m'. 
\]

\noindent
\textbf{Case 3}: The agents in $C_s$ form a component in the envy graph that is a combination of the three configurations and maximal envy chains.   
Suppose that there is $\{i_1,\hdots,i_k\} \subset C_s$ such that agents $i_1,\hdots,i_k$ form a component in the envy graph as described in Figure \ref{Configuration 3}.  Eliminate the arrow from agent $i_k$ to agent $i_1$.
Identify a subset $C'_s  = \{i_1,\hdots,i_k,j_1,\hdots,j_{k'},l\} \subseteq C_s$ such that $C'_s$ form a component of the envy graph that is configured as in either (1) or (2), and adding any agent to $C'_s$ is no longer configured as in any of (1) and (2).  Identify the order on $C'_s$ by following the procedure as described in Case 2.  Update the envy graph as follows: (i) relabel the agents in $C'_s$ in such a way that $g_{i_1}(\omega) < \hdots < g_{i_{k+k'+1}}(\omega)$, and define the chain $i_{k+k'+1} \rightarrow \hdots \rightarrow i_1$.  (ii) For each $i \in C'_s$ and each $j \in C_s \backslash C'_s$, if agent $i$ points to agent $j$ in the original envy graph, agent $i$ still points to agent $j$.  Similarly, for each $i \in C'_s$ and each $j \in C_s \backslash C'_s$, if agent $j$ points to agent $i$ in the original envy graph, agent $j$ still points to agent $i$.   
Repeat the procedure with the updated envy graph.  
Whenever identifying the order on a subset $C'_s$, for each pair of agents who are in $C'_s$ and are not in the same chain, their orders are determined according to their endowments and assignments.   
Hence, no matter which subset of agents we choose, we indeed obtain a unique order on the agent set.
Finally, if the updated envy graph is a chain, identify the order by following the procedure described in Case 1. 
Example \ref{update2} illustrates this procedure.  

\bigskip

\begin{example}\label{update2}
Let $N = \{1,\hdots,6\}$.  Let $\prec \in \mathcal{L}$ be defined by $\omega_1 \prec \omega_2 \prec \hdots \prec \omega_6$.  Let $P \in \mathcal{P}^N$ be a single-peaked preference profile with respect to $\prec$  such that 
\begin{align*}
P_1 &: \omega_4,~\omega_5,\hdots \\
P_2 &: \omega_3,~\omega_4,\hdots \\
P_3 &: \omega_3,\hdots \\
P_4 &: \omega_3,~\omega_2,\hdots \\
P_5 &: \omega_3,~\omega_2,~\omega_1,\hdots \\
P_6 &: \omega_1,\hdots.
\end{align*}
Then 
\[
CR^\omega(P) = (\omega_5,\omega_4,\omega_3,\omega_2,\omega_1,\omega_6). 
\]
The agents in $N$ form an envy graph as described in Figure \ref{example_foralg2}. 
First, eliminate the arrow from agent 5 to agent 3. Similarly, eliminate the arrows from agent 6 to agents 2, 3, and 4, respectively.  
Identify a subset $C' \subset N$ such that $C'$ form a subgraph that is configured as either (1) or (2), and adding any agent to $C'$ is no longer configured as either (1) or (2).  There are two possible cases~(Figure~\ref{example_org2}).
\begin{figure}[H]
\centering
\tikzset{every picture/.style={line width=0.75pt}} 

\begin{tikzpicture}[x=0.75pt,y=0.75pt,yscale=-1.2,xscale=1.2]

\draw [color={rgb, 255:red, 0; green, 0; blue, 0 }  ,draw opacity=0.4 ]   (186.83,181.97) .. controls (218.73,195.33) and (219.46,195.42) .. (248.1,161.94) ;
\draw [shift={(248.97,160.92)}, rotate = 490.53] [color={rgb, 255:red, 0; green, 0; blue, 0 }  ,draw opacity=0.4 ][line width=0.75]    (10.93,-3.29) .. controls (6.95,-1.4) and (3.31,-0.3) .. (0,0) .. controls (3.31,0.3) and (6.95,1.4) .. (10.93,3.29)   ;
\draw [color={rgb, 255:red, 0; green, 0; blue, 0 }  ,draw opacity=0.4 ]   (150.83,136.47) .. controls (180.57,130.13) and (181.78,127.12) .. (204.07,118.63) ;
\draw [shift={(205.83,117.97)}, rotate = 519.44] [color={rgb, 255:red, 0; green, 0; blue, 0 }  ,draw opacity=0.4 ][line width=0.75]    (10.93,-3.29) .. controls (6.95,-1.4) and (3.31,-0.3) .. (0,0) .. controls (3.31,0.3) and (6.95,1.4) .. (10.93,3.29)   ;
\draw [color={rgb, 255:red, 0; green, 0; blue, 0 }  ,draw opacity=0.4 ]   (151.05,146.47) .. controls (181.45,150.33) and (191.46,153.26) .. (206.95,159.3) ;
\draw [shift={(208.66,159.97)}, rotate = 201.5] [color={rgb, 255:red, 0; green, 0; blue, 0 }  ,draw opacity=0.4 ][line width=0.75]    (10.93,-3.29) .. controls (6.95,-1.4) and (3.31,-0.3) .. (0,0) .. controls (3.31,0.3) and (6.95,1.4) .. (10.93,3.29)   ;
\draw [color={rgb, 255:red, 0; green, 0; blue, 0 }  ,draw opacity=0.4 ]   (154.55,141.47) -- (235.33,140.98) ;
\draw [shift={(237.33,140.97)}, rotate = 539.65] [color={rgb, 255:red, 0; green, 0; blue, 0 }  ,draw opacity=0.4 ][line width=0.75]    (10.93,-3.29) .. controls (6.95,-1.4) and (3.31,-0.3) .. (0,0) .. controls (3.31,0.3) and (6.95,1.4) .. (10.93,3.29)   ;

\draw    (145.6,151.8) -- (170.31,166.09) ;
\draw [shift={(172.04,167.1)}, rotate = 210.05] [color={rgb, 255:red, 0; green, 0; blue, 0 }  ][line width=0.75]    (10.93,-3.29) .. controls (6.95,-1.4) and (3.31,-0.3) .. (0,0) .. controls (3.31,0.3) and (6.95,1.4) .. (10.93,3.29)   ;
\draw    (144.35,130.52) -- (168.27,115.2) ;
\draw [shift={(169.95,114.12)}, rotate = 507.36] [color={rgb, 255:red, 0; green, 0; blue, 0 }  ][line width=0.75]    (10.93,-3.29) .. controls (6.95,-1.4) and (3.31,-0.3) .. (0,0) .. controls (3.31,0.3) and (6.95,1.4) .. (10.93,3.29)   ;

\draw    (220,115.97) -- (239.32,128.16) ;
\draw [shift={(240.83,129.47)}, rotate = 221.05] [color={rgb, 255:red, 0; green, 0; blue, 0 }  ][line width=0.75]    (10.93,-3.29) .. controls (6.95,-1.4) and (3.31,-0.3) .. (0,0) .. controls (3.31,0.3) and (6.95,1.4) .. (10.93,3.29)   ;
\draw    (222,163.97) -- (243.93,151.68) ;
\draw [shift={(245.55,150.52)}, rotate = 504.31] [color={rgb, 255:red, 0; green, 0; blue, 0 }  ][line width=0.75]    (10.93,-3.29) .. controls (6.95,-1.4) and (3.31,-0.3) .. (0,0) .. controls (3.31,0.3) and (6.95,1.4) .. (10.93,3.29)   ;

\draw (135.5,133.9) node [anchor=north west][inner sep=0.75pt]    {$6$};
\draw (172,102.4) node [anchor=north west][inner sep=0.75pt]    {$1\ \rightarrow \ 2\ $};
\draw (174.4,162.1) node [anchor=north west][inner sep=0.75pt]    {$5\ \rightarrow \ 4$};
\draw (245.2,133.2) node [anchor=north west][inner sep=0.75pt]    {$3$};

\end{tikzpicture}
\caption{Envy graph formed by the agents in $N$}
\label{example_foralg2}
\end{figure}
\begin{figure}[H]
\centering
\begin{subfigure}{0.49\textwidth}
\centering
\tikzset{every picture/.style={line width=0.75pt}} 

\begin{tikzpicture}[x=0.75pt,y=0.75pt,yscale=-1,xscale=1]

\draw (172,102.4) node [anchor=north west][inner sep=0.75pt]    {$1\ \rightarrow \ 2\ $};
\draw (174,160) node [anchor=north west][inner sep=0.75pt]    {$5\ \rightarrow \ 4$};
\draw (253.5,132.4) node [anchor=north west][inner sep=0.75pt]    {$3$};
\draw (227.44,157.92) node [anchor=north west][inner sep=0.75pt]  [rotate=-330]  {$\rightarrow $};
\draw (234.54,108.42) node [anchor=north west][inner sep=0.75pt]  [rotate=-30]  {$\rightarrow $};

\end{tikzpicture}
\caption{}
\label{example_alg21}
\end{subfigure}
\begin{subfigure}{0.49\textwidth}
\centering
\tikzset{every picture/.style={line width=0.75pt}} 

\begin{tikzpicture}[x=0.75pt,y=0.75pt,yscale=-1,xscale=1]

\draw (127,132.4) node [anchor=north west][inner sep=0.75pt]    {$6$};
\draw (172,102.4) node [anchor=north west][inner sep=0.75pt]    {$1\ \rightarrow \ 2\ $};
\draw (174,160) node [anchor=north west][inner sep=0.75pt]    {$5\ \rightarrow \ 4$};
\draw (150,152) node [anchor=north west][inner sep=0.75pt]  [rotate=-30]  {$\rightarrow $};
\draw (147,118) node [anchor=north west][inner sep=0.75pt]  [rotate=-330]  {$\rightarrow $};

\end{tikzpicture}
\caption{}
\label{example_alg22}
\end{subfigure}
\caption{Possible selection of a subgraph}
\label{example_org2}
\end{figure}

Without loss of generality, let us choose \ref{example_alg21}.  By following the description in Configuration~1, 
\begin{align}\label{update1}
g_3(\omega) < g_2(\omega) < g_1(\omega) < g_4(\omega) < g_5(\omega).
\end{align} 
Update the envy graph according to \eqref{update1}:
\[
6~\rightarrow~5~\rightarrow~4~\rightarrow~1~\rightarrow~2~\rightarrow~3.
\] 
By following the description in Case 1, we obtain 
\[
g_3(\omega) < g_2(\omega) < g_1(\omega) < g_4(\omega) < g_5(\omega) < g_6(\omega).
\]
\end{example}

\bigskip

We show that the crawler induced by an endowment profile and the sequential priority rule induced by the order on the agent set given by the mapping select the same allocation,

\begin{claim}\label{theorem_ equivalence}
For each $\omega \in \mathcal{X}$, 
\[
CR^\omega(P) = SP^{g(\omega)}(P).
\]
\end{claim}

\begin{proof}[Proof of Claim \ref{theorem_ equivalence}]
Suppose that there is a sequence of agents $\{i_1,\cdots,i_k\}~\subseteq~N$ for some $k \in \{1,\hdots,n\}$ such that 
\begin{itemize}
\item for each $k' \in \{1,\cdots,k-1\}$, 
\begin{align*}
SP^{g(\omega)}_{i_{k'}}(P) = CR^\omega_{i_{k'+1}}(P)~P_{i_{k'}}~CR^\omega_{i_{k'}}(P);~\text{and}~
\end{align*} 
\item $CR^\omega_{i_k}(P)~P_{i_k}~SP^{g(\omega)}_{i_{k}}(P)$.
\end{itemize}
This implies that $g_{i_{k-1}}(\omega) < g_{i_k}(\omega)$.  However, because $CR^\omega_{i_k}(P)~P_{i_{k-1}}~CR^\omega_{i_{k-1}}(P)$, there is $r \in \n$ such that agent $i_{k-1}$ envies agent $i_k$ at Round $r$.  This implies that $g_{i_k}(\omega) < g_{i_{k-1}}(\omega)$, contradicting the first inequality. 
\end{proof}

\bigskip

Let $\omega,\omega' \in \mathcal{X}$ be such that if there is a pair $i,j \in N$ such that $\omega'_i = \omega_j$, both agents $i$ and~$j$ belong to either an envy chain or a sugraph of an envy graph that is configured as in (3).  We show that if the allocations are the same for $\omega$ and $\omega'$, the two endowment profiles are the same.  

\begin{claim} \label{shuffle}
Let $\omega,\omega' \in \mathcal{X}$ be such that for each $N' \subseteq N$ such that the agents in $N'$ form either an envy chain or a subgraph of the envy graph that is configured as in (3) at $\omega$, we have $\bigcup_{i \in N'} \omega_i = \bigcup_{i \in N'} \omega'_i$.
If $CR^{\omega'}(P) = CR^\omega(P)$, then $\omega' = \omega$.
\end{claim} 

\begin{proof}[Proof of Claim \ref{shuffle}]
Let $k \in \{1,\hdots,n\}$.
Let $N' = \{1,\cdots,k\} \subseteq N$ be such that at the end of Phase 1,
\begin{itemize}
\item the agents in $N'$ form a chain $1 \rightarrow \hdots \rightarrow k$;\footnote{For each $l \in \{1,\hdots,k-2\}$, agent $l$ may also point to agent $l' \geq l+2$.} and
\item no agent points to agent 1. 
\end{itemize}

Let $\omega,\omega' \in \mathcal{X}$ be such that $CR^{\omega'}(P) = CR^\omega(P)$, $\bigcup_{i \in N'} \omega_i = \bigcup_{i \in N'} \omega'_i$, and if there is a pair $i,j \in N \backslash N'$ such that $\omega_i = \omega'_j$, both agents $i$ and $j$ belong to either an envy chain or a subgraph configured as in (3).
We show that $\omega_{N'} = \omega'_{N'}$.  

Let $O' \equiv \bigcup_{i \in \{1,\hdots,k\}} CR^\omega_{i}(P)$.
Notice that for each $k' \in \{2,\hdots,k-1\}$, there is no $l~\in~\{1,\hdots,k-1\}$ such that $l <k'$ and agent $l$'s assignment is between agents $k'$ and $k'+1$'s assignments.  
Moreover, for each $k' \in \{1,\hdots,k-1\}$, if there is $o~\in~O~\backslash~O'$ such that object $o$ is between agents $k'$ and $k'+1$'s assignments, the agent who receives object $o$ does not envy agent $k'+1$.  
The proof is by induction on $k$.  

Suppose that $k = 1$.  
Because $\bigcup_{i \in N'} \omega_i = \bigcup_{i \in N'} \omega'_i$, we have $\omega_{N'} = \omega'_{N'}$.  

Let $K \in \{2,\cdots,n\}$.  Suppose that for each $k \leq K-1$, we have $\omega_{N'} = \omega'_{N'}$. 
Now we consider $k = K$.  We show that $\omega'_K = \omega_K$.  Then by the \textit{induction hypothesis}, $\omega'_{N'} = \omega_{N'}$.  Let $o \in O$ be such that $CR^\omega_K(P) = o$.  
Because $CR^{\omega'}(P) = CR^\omega(P)$, at both $\omega$ and $\omega'$, for each $k'~\in~\{1,\hdots,K-1\}$, agent $k'+1$ receives his assignment earlier than agent $k'$.  In particular, agent $K$ is the first agent who receives his assignment in $N'$.  This means that there is no $l~\in~\{1,\hdots,K-1\}$ such that agent $l$'s endowment is between object $o$ and agent $K$'s endowment.  

\noindent
\textbf{Case 1}: Either for each $i \in N'$, $\omega_i \precsim o$ or for each $i \in N'$, $o \precsim \omega_i$. 
Because there is no $l \in \{1,\hdots,K-1\}$ such that agent $l$'s endowment is between object $o$ and agent $i_K$'s endowment, $\omega_{K} = \omega'_{K}$.    

\bigskip

Otherwise, there is a pair $i,j \in N'$ such that $\omega_i \precsim o \precsim \omega_j$, and there is no $l \in N' \backslash \{i,j\}$ such that either $\omega_i \prec \omega_l \precsim o$ or $o \precsim \omega_l \prec  \omega_j$.  Let $i,j \in N'$ be such a pair.  

\noindent
\textbf{Case 2}:  Either $i = K$ or $j = K$.  Without loss of generality, suppose that $i = K$. So $\omega_{K}~\precsim~o~\prec~\omega_j$.  
Let $M \subset N$ be such that for each $a \in M$, agent $a$ receives his assignment earlier than agent $K$ when the endowment profile is $\omega$.  
Let 
\begin{align*}
&S_1 \equiv \left\{a \in N:~\omega_{K} \prec \omega_a \prec \omega_j~\text{and}~CR^\omega_a(P) \precsim \omega_{K}\right\}, \\
&T_1 \equiv \left\{a \in N: \omega_{K} \prec \omega_a \prec \omega_j,~CR^{\omega}_a(P) \precsim \omega_a,~\text{and}~\omega_K \prec CR^{\omega}_a(P) \prec o\right\},~\text{and} \\
&U_1 \equiv \left\{a \in N: \omega_{K} \prec \omega_a \prec CR^{\omega}_a(P) \prec o \right\}.\\
\end{align*}
Because $CR^\omega_{K}(P) = o$, we have $|M~\cap(S_1~\cup~T_1~\cup~U_1)| = |o - \omega_{K}|-1$. 
Let 
\[
E_1 = \{o' \in O:~\text{there is}~a \in S_1~\cup~T_1~\cup~U_1~\text{such that}~\omega_a = o'\}. 
\]
Suppose that $\omega'_{K} = \omega_j$.  Let $i' \in \{1,\hdots,K-1\}$ be such that $\omega'_{i'} = \omega_{K}$.  Hence, $\omega'_{i'}~ \precsim o \prec \omega'_{K}$.  
Let
\begin{align*}
&S_2 \equiv \left\{a \in N:~\omega'_{i'} \prec \omega'_a \prec \omega'_{K}~\text{and}~CR^{\omega'}_a(P) \precsim \omega'_{i'}\right\}, \\
&T_2 \equiv \left\{a \in N: \omega'_{i'} \prec \omega'_a \prec \omega'_{K},~CR^{\omega'}_a(P) \precsim \omega_a,~\text{and}~\omega'_{i'} \prec CR^{\omega'}_a(P) \prec o\right\},~\text{and} \\
&U_2 \equiv \left\{a \in N: \omega'_{i'} \prec \omega'_a \prec CR^{\omega'}_a(P) \prec o \right\}.
\end{align*}
Notice that $N'~\cap~(S_1~\cup~T_1~\cup~U_1) = \emptyset$.  This means that for each $a \in N$ such that $\omega'_a \in E_1$, $CR^{\omega'}_a(P) \prec o$.  Otherwise, there are a pair $a,a' \in N'$ and $b \in N \backslash N'$ with either $\omega_b \in E_1$ or $\omega'_b \in E_1$ such that agent $b$ envies agent $a$ and agent $a'$ envies both agents $a$ and $b$.  This means that $b \in N'$, contradicting the fact that $b \notin N'$. 
Hence, $|S_2~\cup~T_2~\cup~U_2| \geq |o - \omega'_{i'}|-1$. 
 
At $\omega'$, the ownership of agent $i'$ can be shifted to object $o$ before agent $K$ is visited.  Hence, if there is $o' \in O'$ such that $o'~P_{i'}~CR^{\omega}_{i'}(P)$ and $o' \precsim o$, we have $CR^{\omega'}_{i'}(P)~P_i~CR^{\omega}_{i'}(P)$, contradicting the hypothesis that $CR^{\omega'}(P) = CR^{\omega}(P)$.  

Suppose that for each $o' \in O'$ such that $o'~P_i~CR^{\omega}_{i'}(P)$, we have $o \prec o'$. 
Without loss of generality, let $m \in \{1,\hdots,K\}$, and $\{o_1,\hdots,o_m\} \subset O$ be such that for each $o'~\in~\{o_1,\hdots,o_m\}$, $o' \in O'$ and $o \prec o_1 \prec \hdots \prec o_m$.  
The fact that agent $K$ receives object $o$ at $\omega'$ implies that the ownership of agent $i'$ is shifted to object $o_1$ earlier than the agent who receives object $o_1$ at $\omega$ is either visited and receives his assignment, or his ownership is shifted to object~$o_1$.  If $o_1~P_{i'}~CR^{\omega}_{i'}(P)$, then $CR^{\omega'}_{i'}(P)~P_{i'}~CR^\omega_{i'}(P)$.  This contradicts the hypothesis that $CR^{\omega'}(P) = CR^{\omega}(P)$.  Suppose that object $o_1$ is assigned to the same agent for both $\omega$~and~$\omega'$.  This means that the ownership of agent $i'$ is shifted to object $o_2$ before the agent who receives object $o_2$ at $\omega$ is either visited and receives his assignment, or his ownership is shifted to object~$o_2$.  If $o_2~P_{i'}~CR^{\omega}_{i'}(P)$, then $CR^{\omega'}_{i'}(P)~P_{i'}~CR^\omega_{i'}(P)$.  This contradicts the hypothesis that $CR^{\omega'}(P) = CR^{\omega}(P)$.  By an analogous argument, $CR^{\omega'}_{i'}(P)~P_i~CR^{\omega}_{i'}(P)$. This contradicts the hypothesis that $CR^{\omega'}(P) = CR^{\omega}(P)$.  Therefore, $\omega'_K = \omega_K$.  

\end{proof}

\bigskip

We are ready to show that mapping $g$ is one-to-one.  

\begin{claim}\label{onetoone}
For each pair $\omega, \omega' \in \mathcal{X}$, 
\[
g(\omega) = g(\omega')~\Longrightarrow~\omega = \omega'.
\]
\end{claim}

\begin{proof}[Proof of Claim \ref{onetoone}]
Let $\omega,\omega' \in \mathcal{X}$ be such that $g(\omega) = g(\omega')$. Hence, $CR^\omega(P) = CR^{\omega'}(P)$.  
Let $i,j \in N$ be such that $\omega'_i = \omega_j$.  
Let $S \in \{1,\hdots,n\}$.  Let $\{C_1,\hdots,C_S\}$ be the partition of $N$ such that for each $s \in \{1,\hdots,S\}$, the agents in $C_s$ form a component in the envy graph.  Let $s \in \{1,\hdots,S\}$ be such that $i \in C_s$.   
Suppose that $j \notin C_s$.  There is $k \in C_s$ such that $g_k(\omega') \neq g_k(\omega)$. 
Hence, $g(\omega) \neq g(\omega')$, contradicting the hypothesis that $g(\omega) = g(\omega')$.  
Suppose that $j \in C_s$.  
Suppose that neither envy chain nor subgraph configured as in (3) to which both agents $i$ and $j$ belong exists. Then there is a pair $k,k' \in C_s$ such that 
\begin{align*}
g_{k}(\omega) < g_{k'}(\omega)~\text{and}~g_k(\omega') > g_{k'}(\omega').  
\end{align*}
Hence, $g(\omega') \neq g(\omega)$, contradicting the hypothesis that $g(\omega) = g(\omega')$. 
Finally, suppose that either an envy chain or a subgraph configured as in (3) to which both agents $i$ and $j$ belong exists.  By Claim \ref{shuffle}, $\omega'_i = \omega_i$.  
\end{proof}

\bigskip

Finally, we show that mapping $g$ is onto. 

\begin{claim} \label{onto}
For each $f \in \mathcal{F}$, there is $\omega \in \mathcal{X}$ such that $g(\omega) = f$.
\end{claim}

\begin{proof}[Proof of Claim \ref{onto}]
By Claims \ref{theorem_ equivalence} and \ref{onetoone}, for each $x \in \mathcal{X}$,
\[\left|\left\{\omega \in \mathcal{X}: CR^\omega(P) = x\right\}\right| \leq \left|\left\{f \in \mathcal{F}: SP^f(P) = x\right\}\right|.\]
Hence,
\[\sum_{x \in \mathcal{X}} \left|\left\{\omega \in \mathcal{X}: CR^\omega(P) = x\right\}\right| \leq \sum_{x \in \mathcal{X}} \left|\left\{f \in \mathcal{F}: SP^f(P) = x\right\}\right|.\]
However, the left-hand side of the inequality is equal to the number of possible endowment profiles and its
right-hand side is equal to the number of orderings over the agents.  Both of these numbers equal to $n!$.  Therefore,
for each $x \in \mathcal{X}$,
\[\left|\left\{\omega \in \mathcal{X}: CR^\omega(P) = x\right\}\right| = \left|\left\{f \in \mathcal{F}: SP^f(P) = x\right\}\right|.\]
This implies that for each $f \in \mathcal{F}$, there is $\omega \in \mathcal{X}$ such that $g(\omega) = f$.
\end{proof}

\section{An example showing that the crawler is not a member of the trading cycles family}\label{example_ap}
For our purpose, it will suffice to formally define the trading cycles family for problems with three agents and three objects.\footnote{Refer to \cite{pu17} for general problems.}
For each $N' = \{i_1,\hdots,i_{n'}\} \subseteq N$ and each $O' \subseteq O$ such that $|N'| = |O'|$, a ``partial allocation for $N'$" is a list $x_{N'} = (x_{i_1},\hdots,x_{i_{n'}})$ such that for each $i \in N'$, $x_i \in O'$, and for each pair $i,j \in N'$ such that $i \neq j$, $x_i \neq x_j$.  We call each $i \in N \backslash N'$ an ``unassigned agent", and each $o \in O \backslash O'$ an ``unassigned object".  Let $\mathcal{X}_{N'}$ be the set of partial allocations for $N'$.   Let $\mathcal{Y} = \bigcup_{N' \subseteq N} \mathcal{X}_{N'}$. Obviously, $\mathcal{X} \subset \mathcal{Y}$.  Let $\overline{\mathcal{Y}} = \mathcal{Y}~\backslash~\mathcal{X}$.  For each $y \in \mathcal{Y}$, let $\overline{N}_y$ be the set of unassigned agents at $y$ and $\overline{O}_y$ be the set of unassigned objects at $y$. 
 
A ``control-rights structure" is a collection of mappings 
\[\left\{\left(c_y,b_y\right):~\overline{O}_y \rightarrow \overline{N}_y \times \left\{\text{ownership},\text{brokerage}\right\}\right\}_{y \in \overline{\mathcal{Y}}}\]
such that 
\begin{itemize}\itemsep-0.2cm
\item[(R1)] At the empty allocation, i.e., when the set of unassigned objects is the entire set of objects, each object is brokered.\footnote{This assumption is only for problems with three agents and three objects. It differs for general problems.}
\item[(R2)] If agent $i$ is the only agent who has not been assigned an object at $y$, he owns all of the unassigned objects at $y$. 
\item[(R3)] If agent $i$ brokers an object at $y$, he does not control any other objects at $y$. 
\end{itemize}   
Moreover, for each pair $y,y' \in \overline{\mathcal{Y}}$ such that $y \subset y'$ and each $i \in \overline{N}_{y'}$ such that there is $o \in \overline{O}_{y'}$ that agent $i$ owns at $y$,\footnote{$y \subset y'$ means that if an agent is assigned an object at $y$, then he is assigned the same object at $y'$.} 
\begin{itemize}\itemsep-0.2cm
\item[(R4)] Agent $i$ owns object $o$ at $y'$.
\item[(R5)] If agent $j \in \overline{N}_{y'}$ brokers object $o' \in \overline{O}_{y'}$ at $y$, then either agent $j$ brokers object $o'$ at $y'$ or agent $i$ owns object $o'$ at $y'$.
\item[(R6)] If agent $j \in \overline{N}_{y'}$ controls object $o' \in \overline{O}_{y'}$ at $y$, then agent $j$ owns object $o$ at $y~\cup~\{(i,o')\}$.\footnote{$(i,o')$ means that agent $i$ is assigned object $o'$.} 
\end{itemize}

\paragraph{Trading Cycles:}(For problems with three agents and three objects.) Let $(c,b)$ be a control-rights structure.  Let $P \in \mathcal{P}^N$. 

Let $x^0$ be the empty allocation.  In each Round $r = 1,2,\hdots,$ some agents are assigned objects.  We denote by $x^{r-1}$ the partial allocation at Round $r-1$.  Note that $x^{r-1} \subset x^r$.  If $x^{r-1} \in \overline{Y}$, the algorithm proceeds with the following steps of Round $r$. 

\noindent\textbf{Step 1}: \textbf{Pointing}. Each $o \in \overline{O}_{x^{r-1}}$ points to the agent who controls it at $x^{r-1}$.  Each $i \in \overline{N}_{x^{r-1}}$ points to his most preferred object in $\overline{O}_{x^{r-1}}$.  

\bigskip

A cycle 
\[o^1 \rightarrow i^1 \rightarrow \hdots o^k \rightarrow i^k \rightarrow o^1\]
is simple if there is an agent in the cycle who owns an object. 

\noindent\textbf{Step 2$\bm{a}$}: \textbf{Executing ``simple" cycles}.  Each agent in each simple cycle is assigned the object he is pointing to.

\noindent\textbf{Step 2$\bm{b}$}: \textbf{Forcing brokers to downgrade their pointing}.  If there are non-simple cycles, we proceed as follows:
\begin{itemize}\itemsep-0.2cm
\item If there is a cycle in which there is a broker $i$ who points to a brokered object, and there is another broker or owner who points to the same object, we force broker $i$ to point to his next most preferred object in $\overline{O}_{x^{r-1}}$ and we return to Step 2$(a)$.
\item Otherwise, we clear all cycles by assigning each agent in each cycle the object he is pointing to. 
\end{itemize}

\noindent\textbf{Step 3}: Each agent who has received his assignment leaves with his assignment. When each agent received his assignment, the algorithm terminates. 

The following example is a problem for which there is no control-rights structure such that the associated rule always select the same allocation as the crawler.  

\begin{example}\label{crawler_pu}
Let $N = \{1,2,3\}$ and $O = \{o_1,o_2,o_3\}$.  Let $\omega = (o_1,o_2,o_3)$.  Let $P \in \mathcal{P}^N$ be such that 
\begin{align*}
P_1 &:~o_2,\hdots \\
P_2 &:~o_1,~o_2,~o_3 \\
P_3 &:~o_1,~o_2,~o_3.
\end{align*}
Then
\begin{align*}
CR^\omega(P) = (o_2,o_1,o_3).  
\end{align*}
There are two control-rights structure such that $TC^{(c,b)}(P) = CR^\omega(P)$;
\begin{enumerate}\itemsep-0.2cm
\item[\textbf{Case $\bm{1}$}]: At the empty allocation, agents 1, 2, and 3 broker objects $o_1$, $o_3$, and $o_2$, respectively.
\item[\textbf{Case $\bm{2}$}]: At the empty allocation, agents 1, 2, and 3 broker objects $o_3$, $o_2$, and $o_1$, respectively.  
\end{enumerate}
Consider Case 1. For each $i \in N$, let $P'_i \in \mathcal{P}^N$ be defined by 
\begin{align*}
P'_i:~o_1,~o_2,~o_3.
\end{align*}
Then 
\begin{align*}
CR^\omega(P') = (o_1,o_2,o_3) \neq (o_2,o_1,o_3) = TC^{(c,b)}(P').  
\end{align*}
Consider Case 2.  For each $i \in N$, let $\tilde{P} \in \mathcal{P}^N$ be defined by 
\begin{align*}
\tilde{P}_i:~o_3,~o_2,~o_1.
\end{align*}
Then 
\begin{align*}
CR^\omega(\tilde{P}) = (o_1,o_2,o_3) \neq (o_2,o_1,o_3) = TC^{(c,b)}(\tilde{P}). 
\end{align*}
\end{example}
\end{appendices}

\newpage
\bibliographystyle{agsm}
\bibliography{Crawler}

@Article{as98,
  Title                    = {Random serial dictatorship and the core from random endowments in house allocation problems},
  Author                   = {Abdulkadiroglu, Atila and S\"{o}nmez, Tayfun},
  Journal                  = {Econometrica},
  Year                     = {1998},
  Number                   = {3},
  Pages                    = {689--701},
  Volume                   = {66},

  Publisher                = {Blackwell Publishing Ltd.}
}

@Article{anno15,
  Title                    = {A short proof for the characterization of the core in housing markets},
  Author                   = {Anno, Hidekazu},
  Journal                  = {Economics Letters},
  Year                     = {2015},
  Pages                    = {66--67},
  Volume                   = {126},

  Publisher                = {Elsevier}
}

@Article{bade19equiv,
  Title                    = {Random serial dictatorship: the one and only},
  Author                   = {Bade, Sophie},
  Journal                  = {Mathematics of Operations Research},
  Year                     = {2020},
  Number                   = {1},
  Pages                    = {353--368},
  Volume                   = {45},

  Publisher                = {INFORMS}
}

@Article{bade19,
  Title                    = {Matching with single-peaked preferences},
  Author                   = {Bade, Sophie},
  Journal                  = {Journal of Economic Theory},
  Year                     = {2019},
  Pages                    = {81--99},
  Volume                   = {180},

  Publisher                = {Elsevier}
}

@Article{bmrs19,
  author  = {Beynier, Aur{\'e}lie and Maudet, Nicolas and Rey, Simon and Shams, Parham},
  journal = {arXiv preprint},
  title   = {House Markets and Single-Peaked Preferences: From Centralized to Decentralized Allocation Procedures},
  year    = {2019},
}

@Article{carroll14,
  Title                    = {A general equivalence theorem for allocation of indivisible objects},
  Author                   = {Carroll, Gabriel},
  Journal                  = {Journal of Mathematical Economics},
  Year                     = {2014},
  Pages                    = {163--177},
  Volume                   = {51},

  Publisher                = {Elsevier}
}

@InProceedings{dbcm15,
  Title                    = {The power of swap deals in distributed resource allocation},
  Author                   = {Damamme, Anastasia and Beynier, Aur{\'e}lie and Chevaleyre, Yann and Maudet, Nicolas},
  Booktitle                = {Proceedings of the 2015 International Conference on Autonomous Agents and Multiagent Systems},
  Year                     = {2015},
  Organization             = {International Foundation for Autonomous Agents and Multiagent Systems},
  Pages                    = {625--633}
}

@Article{ekici17,
  author  = {Ekici, \"{O}zg\"{u}n},
  journal = {Journal of Mathematical Economics},
  title   = {Random mechanisms for house allocation with existing tenants},
  year    = {2020},
  pages   = {53-65},
  volume  = {89},
}

@InProceedings{elo08,
  Title                    = {Single-peaked consistency and its complexity},
  Author                   = {Escoffier, Bruno and Lang, J{\'e}r{\^o}me and {\"O}zt{\"u}rk, Meltem},
  Booktitle                = {18th European Conference on Artificial Intelligence},
  Year                     = {2008},
  Pages                    = {366--370},
  Volume                   = {8}
}

@Article{knuth96,
  Title                    = {An exact analysis of stable allocation},
  Author                   = {Knuth, Donald E},
  Journal                  = {Journal of Algorithms},
  Year                     = {1996},
  Number                   = {2},
  Pages                    = {431--442},
  Volume                   = {20},

  Publisher                = {Elsevier}
}

@Article{ls11,
  Title                    = {Equivalence results in the allocation of indivisible objects: A unified view},
  Author                   = {Lee, Thiam and Sethuraman, Jay},
  Journal                  = {Columbia University, mimeo},
  Year                     = {2011}
}

@Article{liu18,
  Title                    = {A large class of strategy-proof exchange rules with single-peaked preferences},
  Author                   = {Liu, Peng},
  Journal                  = {Research Collection School of Economics},
  Year                     = {2018}
}

@Article{ma94,
  Title                    = {Strategy-proofness and the strict core in a market with indivisibilities},
  Author                   = {Ma, Jinpeng},
  Journal                  = {International Journal of Game Theory},
  Year                     = {1994},
  Number                   = {1},
  Pages                    = {75--83},
  Volume                   = {23},

  Publisher                = {Springer}
}

@Article{ps11,
  Title                    = {Lotteries in student assignment: An equivalence result},
  Author                   = {Pathak, Parag A and Sethuraman, Jay},
  Journal                  = {Theoretical Economics},
  Year                     = {2011},
  Number                   = {1},
  Pages                    = {1--17},
  Volume                   = {6},

  Publisher                = {Wiley Online Library}
}

@Article{pu17,
  Title                    = {Incentive compatible allocation and exchange of discrete resources},
  Author                   = {Pycia, Marek and {\"U}nver, M Utku},
  Journal                  = {Theoretical Economics},
  Year                     = {2017},
  Number                   = {1},
  Pages                    = {287--329},
  Volume                   = {12},

  Publisher                = {Wiley Online Library}
}

@Article{sethuraman16,
  Title                    = {An alternative proof of a characterization of the {TTC} mechanism},
  Author                   = {Sethuraman, Jay},
  Journal                  = {Operations Research Letters},
  Year                     = {2016},
  Number                   = {1},
  Pages                    = {107--108},
  Volume                   = {44},

  Publisher                = {Elsevier}
}

@Article{ss74,
  Title                    = {On cores and indivisibility},
  Author                   = {Shapley, Lloyd and Scarf, Herbert},
  Journal                  = {Journal of Mathematical Economics},
  Year                     = {1974},
  Number                   = {1},
  Pages                    = {23--37},
  Volume                   = {1},

  Publisher                = {Elsevier}
}

@Article{sprumont91,
  author    = {Sprumont, Yves},
  journal   = {Econometrica},
  title     = {The division problem with single-peaked preferences: a characterization of the uniform allocation rule},
  year      = {1991},
  number    = {2},
  pages     = {509--519},
  volume    = {59},
  publisher = {JSTOR},
}

@Article{svensson99,
  Title                    = {Strategy-proof allocation of indivisible goods},
  Author                   = {Svensson, Lars-Gunnar},
  Journal                  = {Social Choice and Welfare},
  Year                     = {1999},
  Number                   = {4},
  Pages                    = {557--567},
  Volume                   = {16},

  Publisher                = {Springer}
}

@Article{tamura21,
  author  = {Tamura, Yuki},
  journal = {mimeo},
  title   = {Object Reallocation Problems under Single-Peaked Preferences: Two Characterizations of The Crawler},
  year    = {2021},
}

@Article{su05,
  author    = {S{\"o}nmez, Tayfun and {\"U}nver, M Utku},
  journal   = {Games and Economic Behavior},
  title     = {House allocation with existing tenants: an equivalence},
  year      = {2005},
  number    = {1},
  pages     = {153--185},
  volume    = {52},
  publisher = {Elsevier},
}
\end{document}